\documentclass[11pt]{article}

\usepackage[margin=1in]{geometry}
\usepackage{parskip}
\usepackage{amsmath, amssymb, amsthm}
\usepackage{booktabs}
\usepackage{hyperref}
\usepackage{natbib}
\usepackage{xcolor}
\usepackage{multirow}
\usepackage{tikz}
\usetikzlibrary{positioning, arrows.meta, shapes.geometric, fit, backgrounds, calc}

\newtheorem{theorem}{Theorem}

\newcommand{\E}{\mathbb{E}}
\newcommand{\Var}{\mathrm{Var}}
\newcommand{\Cov}{\mathrm{Cov}}
\newcommand{\Prob}{\mathbb{P}}
\newcommand{\R}{\mathbb{R}}
\newcommand{\ind}{\mathbf{1}}
\newcommand{\bfbeta}{\boldsymbol{\beta}}
\newcommand{\bfX}{\mathbf{X}}
\newcommand{\bfXs}{\mathbf{X}^*}

\title{Bias-corrected Cox regression with AI-extracted covariates via calibration summary statistics}

\usepackage{authblk}

\author{Arjun Sondhi}
\affil{Northwell Health}

\date{\today}

\begin{document}
\maketitle

\begin{abstract}
Large-scale observational studies increasingly rely on AI pipelines to extract structured variables from unstructured clinical records. 
A common workflow separates the data vendor, who validates extraction accuracy with a gold-standard sample, from the downstream researcher, who receives only the extracted dataset and summary accuracy statistics. 
We develop a bias-correction framework for the Cox proportional hazards model when covariates are subject to AI extraction error. 
Within a unified multivariate calibration framework, we show that the naive Cox estimator's bias decomposes into a leading-order calibration term and a second-order residual that vanishes as extraction accuracy improves. 
The leading-order term yields a corrected estimator that operates as a post-hoc matrix multiplication on the output of any standard Cox software. 
We further derive bias-adjusted confidence intervals that incorporate calibration uncertainty and a sensitivity diagnostic for assessing whether the neglected residual could materially affect inference. 
Synthetic data experiments with cross-dependent extraction errors and controlled nonlinear calibration violations confirm that the correction substantially reduces bias and achieves near-nominal coverage even under mild violations of the linear calibration assumption. 
The framework yields a concrete reporting specification: a short list of summary statistics that data vendors should provide alongside any AI-extracted covariate dataset used in survival analysis.
\end{abstract}

\section{Introduction}
\label{sec:intro}

The growing availability of electronic health records, clinical notes, and registry databases has enabled observational studies at unprecedented scale. 
Increasingly, AI-based extraction tools---including large language models (LLMs) and other natural language processing (NLP) pipelines---automate the extraction of structured analytic variables from unstructured source data, and a growing number of commercial data vendors offer AI-extracted clinical datasets as research products \citep{adamson2023approach, agrawal2023development}.

A common workflow has emerged: a data vendor develops an AI extraction pipeline, validates it against manually abstracted gold-standard labels, and delivers the extracted dataset to end-user researchers. 
The researchers receive the extracted variables but \emph{not} the paired validation data.
Instead, what the vendors often provide alongside the extracted data is a set of summary statistics characterizing extraction accuracy.
These often include sensitivity and specificity for categorical variables, while continuous variable error summaries may include mean absolute error or root mean squared error \citep{estevez2026ensuring}.
While these metrics are useful for broadly assessing the accuracy of the AI extraction task, they do not directly address how reliable downstream statistical analyses are when applied to the AI-extracted dataset. 

It is well established that using error-prone variables in statistical models (whether from manual abstraction mistakes, instrument noise, or algorithmic extraction) can produce systematic distortions in parameter estimates.
AI-extracted covariates introduce a modern instance of this classical problem: the extraction algorithm acts as a noisy measurement device, and the resulting surrogates inherit error patterns that are shaped by the algorithm's architecture, training data, and the complexity of the source text.
We focus on the Cox model because of its central role in clinical research: it is the default analysis model for time-to-event endpoints, and regulatory submissions routinely rely on Cox-based hazard ratio estimates for efficacy evaluation.
The central question is:
\begin{quote}
\emph{Given accuracy summaries reported by the data vendor, can a downstream researcher correct for the bias in a Cox proportional hazards model fit to AI-extracted covariates and conduct valid inference on the corrected estimates?}
\end{quote}

This paper focuses on the setting where covariate extraction is the primary source of measurement error: the event time $T$ and event indicator $\Delta$ are assumed to be observed correctly (e.g., because they come from administrative or registry sources that require no language extraction). 

Our main contributions are as follows. 
First, we derive a closed-form bias decomposition (Theorem~\ref{thm:cov_bias}) that expresses the difference $\bfbeta - \bfbeta^*$ between the true Cox coefficients and the probability limit of the naive estimator in terms of a \emph{leading-order calibration correction} and a \emph{second-order regression calibration (RC) residual}. 
The leading-order term depends only on the multivariate calibration slope matrix $\mathbf{B}$ and the naive estimate $\hat{\bfbeta}^*$, both of which are available to the researcher, yielding a fully computable corrected estimator $\hat{\bfbeta}_{\mathrm{corr}} = (\hat{\mathbf{B}}^\top)^{-1}\hat{\bfbeta}^*$.
Second, we develop a framework for bias-adjusted inference (Section~\ref{sec:bias_ci}) that propagates both sampling variability in $\hat{\bfbeta}^*$ and estimation uncertainty in $\hat{\mathbf{B}}$ through the correction. 
Crucially, the correction is \emph{analysis-flexible}: it operates on the output of any standard Cox software, requiring no modifications to the partial likelihood fitting procedure itself, and accommodates arbitrary downstream contrasts (e.g., individual coefficients, linear combinations, or joint tests).
Third, the bias decomposition cleanly separates quantities into those reported by the vendor and those computed by the researcher, yielding a concrete \emph{reporting specification} of which accuracy summaries should accompany any AI-extracted covariate dataset used in survival analysis (Table~\ref{tab:stats}).

The rest of this paper is organized as follows. Section~\ref{sec:related} reviews the measurement error, post-prediction inference, and real-world data quality assurance literature, motivating the need for our work. 
Section~\ref{sec:bias} develops the methodology: Section~\ref{sec:setup} fixes notation and states the working assumptions, Section~\ref{sec:bias_cov} presents the bias decomposition of Theorem~\ref{thm:cov_bias} and the resulting corrected estimator, and Section~\ref{sec:bias_ci} derives plug-in and propagated-variance confidence intervals along with a sensitivity diagnostic for the neglected residual. 
Section~\ref{sec:sims} reports simulation studies under cross-dependent extraction errors (Section~\ref{sec:cross-dependent}) and under controlled violations of the linear calibration assumption at constant total error (Section~\ref{sec:nonlinear}).
Section~\ref{sec:discussion} concludes with a discussion of the vendor reporting standard and the framework's limitations.
Appendix~\ref{app:proofs} contains the proof of Theroem~\ref{thm:cov_bias}, Appendix~\ref{app:prop-ci-cvg} examines how confidence interval coverage converges with the validation sample size, and Appendix~\ref{app:mimic} presents a semi-synthetic analysis of MIMIC-IV data.

\subsection*{Generalizable Insights about Machine Learning in the Context of Healthcare}

As AI extraction pipelines become standard infrastructure for large-scale clinical research, a recurring structural problem emerges: the team that builds and validates the extraction model is typically not the team that conducts the downstream analysis, and the validation data rarely travel with the extracted dataset. 
\begin{enumerate}
    \item  We show that given appropriately chosen summary statistics from the validation stage, a downstream analyst can recover bias-corrected point estimates and valid confidence intervals from standard software output, without access to paired validation data or specialized estimation procedures.
    \item The approach establishes a broader design principle: for any downstream analysis where ML-generated variables feed into a formal statistical procedure, one can ask what minimal summary statistics from the ML validation stage are sufficient for inferential correction, and design reporting standards around that answer.
    \item We believe this shift from reporting extraction accuracy as an end in itself to reporting quantities that enable valid downstream inference is a principled design framework for the growing ecosystem of AI-extracted clinical research data.
\end{enumerate}

\section{Related Work}
\label{sec:related}

\paragraph{Measurement error in regression models.}
The statistical theory of measurement error is mature, with comprehensive treatments in the monographs of \citet{fuller2009} and \citet{carroll2006}. 
For linear models, additive covariate error with known variance produces the classical attenuation factor $\sigma_X^2 / (\sigma_X^2 + \sigma_U^2)$, and regression calibration provides a consistent correction. 
For generalized linear models, the induced bias depends on the link function, and regression calibration is only approximately consistent, with the quality of the approximation depending on the curvature of the link and the magnitude of the measurement error \citep{carroll2006}.

\paragraph{Measurement error in Cox models.}
Regression calibration for the Cox model was formalized by \citet{wang1997rc}, who showed that replacing the error-prone covariate with its conditional expectation given the surrogate is only approximately consistent because the induced hazard function does not satisfy proportional hazards; \citet{xie2001} proposed a risk-set calibration variant that recalibrates at each event time. 
Simulation extrapolation (SIMEX) was adapted to survival analysis and extended to handle mixtures of Berkson and classical errors \citep{tapsoba2019simex, oh2018considerations}. 
Bayesian approaches have also been proposed, using parametric models for the baseline hazard to enable full posterior inference over the measurement error structure \citep{bartlett2018bayesian}.
\citet{wangsong2021} provided a unified semiparametric regression calibration framework for general hazard models, characterizing the RC inconsistency in terms of the calibration residual variance and the baseline hazard. 
A common thread across all of these methods is that the researcher is assumed to have direct access to either a validation dataset, replicate measurements, or a known error distribution.

\paragraph{Post-prediction inference.}
A separate and more recent line of work addresses the problem of conducting valid statistical inference when some variables (typically outcomes) are predicted by a machine learning model rather than directly observed. 
\citet{wang2020postpi} formalized the \emph{post-prediction inference} problem, showing that naively substituting predicted outcomes into regression models produces biased estimates and anticonservative standard errors, and proposed corrections based on modeling the relationship between predicted and true outcomes in a labeled testing set. 
This was applied in the Cox regression setting by \citet{sondhi2023postprediction}.
\citet{angelopoulos2023ppi} extended this to \emph{prediction-powered inference} (PPI), a framework that provides valid inference for any estimand defined through an estimating equation. 
\citet{miao2025assumption} proposed the PoSt-Prediction Adaptive (PSPA) inference framework, which achieves element-wise variance reduction relative to using only the labeled data, and established connections to semiparametric efficiency theory. 
These methods assume the researcher has access to a labeled dataset in which both the true and predicted values are jointly observed; the correction is computed from this paired data. 
Moreover, the methods have been developed primarily for estimands defined through smooth M-estimation problems (means, quantiles, GLM coefficients), and their extension to the Cox partial likelihood (which involves risk-set-dependent weights and a nonparametric baseline hazard) is not immediate. 

\paragraph{Quality-assurance frameworks for AI-extracted clinical data.}
A parallel applied literature specifies how vendors should certify RWD quality.
Checklists such as PALISADE \citep{padula2022machine} and SUITABILITY \citep{fleurence2024assessing} target
transparency and fitness-for-purpose, but were not built specifically for extracted data quality.
Closest to our setting is the VALID framework \citep{estevez2026ensuring}, which validates LLM-extracted oncology datasets through three pillars: variable-level accuracy metrics; verification checks for conformance, plausibility, and internal consistency; and replication analyses comparing an analysis run on extracted data against a human-abstracted reference.
VALID is the most complete such specification we are aware of, but fails to yield a statement about the reliability of a downstream effect estimate.
Its accuracy metrics are marginal, and marginal metrics are not sufficient statistics for Cox bias: two pipelines with identical recall, precision, and $F_1$ can induce different calibration slope matrices and hence different asymptotic bias in $\hat{\bfbeta}^*$ (Theorem~\ref{thm:cov_bias}). 
Agreement on a marginal survival curve, as in VALID's case study, does not certify adjusted coefficients in a multivariable Cox model, where extraction errors across several covariates interact.
Verification checks are diagnostic rather than corrective.
The replication pillar is the most inferentially ambitious, yet it requires the vendor to hold the paired data and to run the analysis, so its guarantee does not transport to the analyses a downstream researcher actually specifies.

\section{Estimating Cox Model Bias from Summary Statistics}
\label{sec:bias}

\subsection{Setup and Notation}
\label{sec:setup}

Let $T$ denote the event time, $\Delta$ the event indicator ($\Delta = 1$ if the event is observed, $\Delta = 0$ if censored), and $\bfX = (X_1, \ldots, X_p)^\top$ a $p$-vector of covariates.
$T$ and $\Delta$ are observed exactly; the covariates are observed only through AI-extracted surrogates $\bfXs = (X_1^*, \ldots, X_p^*)^\top$. The true Cox model is
\begin{equation}
    \lambda(t \mid \bfX) = \lambda_0(t)\exp(\bfbeta^\top \bfX),
    \label{eq:cox}
\end{equation}
and $\bfbeta$ is estimated by maximizing the partial likelihood
\begin{equation*}
    \mathrm{PL}(\bfbeta) = \prod_{i:\Delta_i = 1} \frac{\exp(\bfbeta^\top \bfX_i)}{\sum_{j \in \mathcal{R}(T_i)} \exp(\bfbeta^\top \bfX_j)},
\end{equation*}
where $\mathcal{R}(t) = \{j : T_j \ge t\}$ is the risk set at time $t$. We write $\hat{\bfbeta}^*$ for the naive estimator obtained by substituting $\bfXs$ for $\bfX$ in the partial likelihood while keeping $T$ and $\Delta$ at their true values, and $\bfbeta^*$ for its asymptotic target. The bias is $\mathbf{b} := \bfbeta - \bfbeta^*$.

\paragraph{Vendor/researcher separation.} The data vendor has access to a validation dataset of $n_v$ paired observations $(\bfX_i, \bfXs_i)$ (with outcomes, when available, serving only as stratification variables for diagnostics) and reports accuracy summaries derived from it. The researcher receives only the extracted data $(T_i, \Delta_i, \bfXs_i)$ and these summaries. We use ``vendor-reported'' and ``researcher-computed'' throughout to indicate provenance.

\paragraph{Working assumptions.} Unless otherwise stated, we assume:
\begin{enumerate}
    \item[(A1)] \emph{Cox data-generating model and regularity}: The Cox model \eqref{eq:cox} is correctly specified for the true data, and censoring is independent of the event time conditional on the covariates. Additionally, we assume the following regularity conditions: (i) the true coefficient $\bfbeta$ lies in a compact set $\mathcal{B} \subset \R^p$; (ii) $(\bfX, \bfXs)$ is jointly sub-Gaussian, the moment generating function $\E[\exp(\bfbeta^\top\bfX)]$ is uniformly bounded on the compact set $\mathcal{B}$, and all polynomial moments of $\bfX$ and $\bfXs$ are finite; (iii) the study horizon $\tau < \infty$ is finite, the baseline hazard $\lambda_0$ and censoring survival $G(\cdot \mid \bfX)$ are bounded on $[0, \tau]$, and $s_X^{(0)}(\bfbeta, t) = \E[\exp(\bfbeta^\top\bfX)\ind(T \ge t)]$ is bounded below uniformly in $t \in [0, \tau]$; (iv) the Fisher information is finite and the usual second-moment conditions hold.
    \item[(A2)] \emph{Non-differential covariate extraction error}: $\bfXs \perp (T, \Delta) \mid \bfX$. The AI extracts covariates based only on the source record's covariate content, not on downstream outcome information.
    \item[(A3)] \emph{Conditional linear calibration}: the multivariate calibration function is linear: $\E[\bfX \mid \bfXs = \mathbf{x}^*] = \mathbf{a} + \mathbf{B}\mathbf{x}^*$, where $\mathbf{B} = \Cov(\bfX, \bfXs)[\Var(\bfXs)]^{-1}$ is the $p \times p$ calibration slope matrix and $\mathbf{a} = \E[\bfX] - \mathbf{B}\E[\bfXs]$. Moreover, $\mathbf{B}$ is nonsingular, and its condition number $\kappa(\mathbf{B})$ is bounded.
\end{enumerate}
(A2) is a standard assumption for characterizing bias due to measurement error. 
Note that no assumption of independent extraction errors \emph{across covariates} is needed: the unified calibration matrix $\mathbf{B}$ in Section~\ref{sec:bias_cov} captures all cross-covariate correlations automatically.
(A3) holds exactly when $(\bfX, \bfXs)$ is jointly Gaussian; it is most defensible for continuous covariates and only approximate for binary ones.
Crucially, (A3) is testable on the validation data via the goodness-of-fit of the calibration regression, so vendors can flag the magnitude of violations.

Figure~\ref{fig:workflow} summarizes the methodology. 
The vendor provides the multivariate calibration slope $\hat{\mathbf{B}}$ from a validation dataset comparing AI-extracted covariates to gold-standard labels.
The researcher uses $\hat{\mathbf{B}}$ together with the naive Cox fit on the extracted data to estimate and correct the measurement-error bias.
Additional summary statistics and computations yield bias-corrected statistical inference (Section~\ref{sec:bias_ci}).

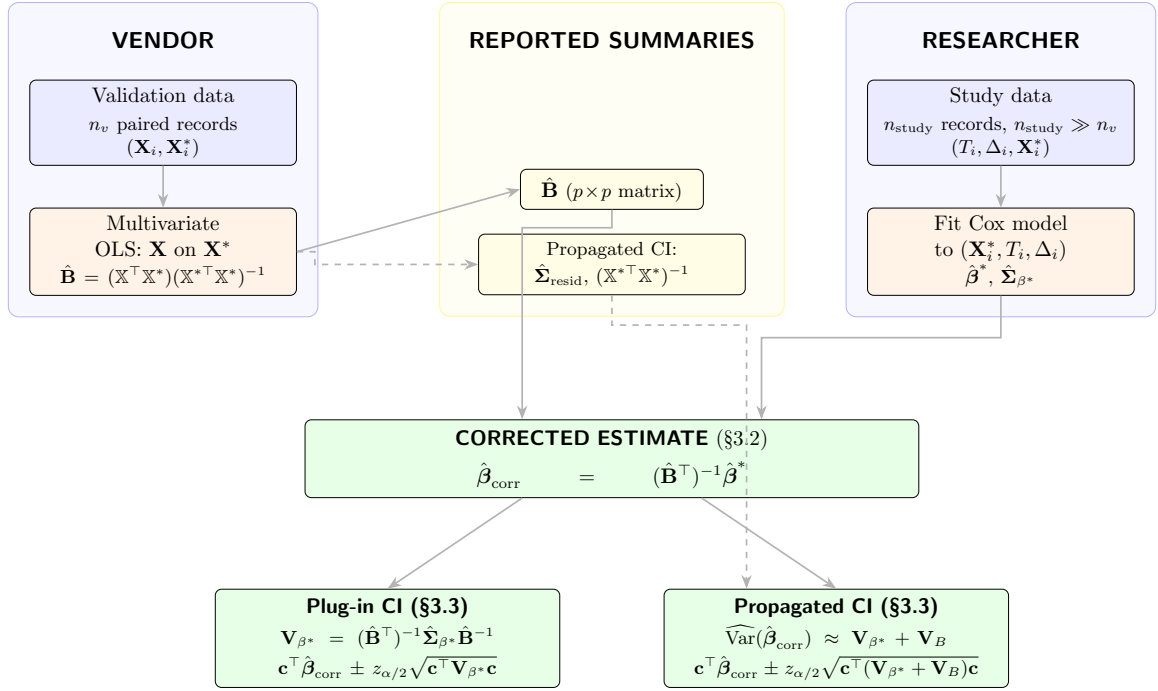
\begin{figure}[!ht]
\centering
\resizebox{0.92\textwidth}{!}{%
\begin{tikzpicture}[
    >=Stealth,
    box/.style={draw, rounded corners=3pt, minimum height=0.8cm, align=center, font=\small, text width=4.2cm},
    databox/.style={box, fill=blue!8},
    compbox/.style={box, fill=orange!10},
    resultbox/.style={box, fill=green!10},
    summbox/.style={draw, rounded corners=3pt, fill=yellow!12, minimum height=0.65cm, align=center, font=\small, text width=2.8cm},
    cibox/.style={draw, rounded corners=3pt, fill=green!10, minimum height=1.1cm, align=center, font=\small, text width=5.5cm},
    lbl/.style={font=\normalsize\sffamily\bfseries},
    arr/.style={->, thick, gray!60},
    darr/.style={->, thick, gray!60, dashed}
]

\node[lbl] (vlbl) at (0, 0) {VENDOR};
\node[lbl] (slbl) at (7.5, 0) {REPORTED SUMMARIES};
\node[lbl] (rlbl) at (14, 0) {RESEARCHER};

\node[databox, below=0.4cm of vlbl] (vdata) {Validation data\\[1pt]{\footnotesize $n_v$ paired records}\\[-1pt]{\footnotesize $(\bfX_i, \bfXs_i)$}};

\node[databox, below=0.4cm of rlbl] (rdata) {Study data\\[1pt]{\footnotesize $n_{\text{study}}$ records, $n_{\text{study}} \gg n_v$}\\[-1pt]{\footnotesize $(T_i, \Delta_i, \bfXs_i)$}};

\node[compbox, below=0.7cm of vdata] (compB) {Multivariate OLS: $\bfX$ on $\bfXs$\\[1pt]{\footnotesize $\hat{\mathbf{B}} = (\mathbb{X}^\top\mathbb{X}^*)(\mathbb{X}^{*\top}\mathbb{X}^*)^{-1}$}};
\draw[arr] (vdata) -- (compB);

\node[summbox] at (7.5, -2.5) (summ1) {$\hat{\mathbf{B}}$ {\footnotesize ($p\!\times\!p$ matrix)}};
\node[summbox, below=0.4cm of summ1, text width=4.2cm] (summOpt) {{\footnotesize Propagated CI:}\\[1pt]{\footnotesize $\hat{\boldsymbol{\Sigma}}_{\mathrm{resid}}$, $(\mathbb{X}^{*\top}\mathbb{X}^*)^{-1}$}};

\draw[arr] (compB.east) -- (summ1.west);
\draw[darr] (compB.east) -- ++(0.3,0) |- (summOpt.west);

\node[compbox, below=0.7cm of rdata] (compCox) {Fit Cox model to $(\bfXs_i, T_i, \Delta_i)$\\[1pt]{\footnotesize $\hat{\bfbeta}^*$, $\hat{\boldsymbol{\Sigma}}_{\beta^*}$}};

\draw[arr] (rdata) -- (compCox);

\node[resultbox, text width=10cm, minimum height=1.3cm] at (7.5, -7) (corr) {%
\textsf{\textbf{CORRECTED ESTIMATE}} (\S\ref{sec:bias_cov})\\[4pt]
$\hat{\bfbeta}_{\mathrm{corr}} \;=\; (\hat{\mathbf{B}}^\top)^{-1}\hat{\bfbeta}^*$
};

\draw[arr] (compCox.south) -- ++(0,-0.7) -| ([xshift=2.5cm]corr.north);
\draw[arr] (summ1.south) -- ++(0,-0.3) -| ([xshift=-1.5cm]corr.north);

\node[cibox] at (3.75, -10) (plugci) {%
\textsf{\textbf{Plug-in CI (\S\ref{sec:bias_ci}) }}\\[3pt]
{\footnotesize $\mathbf{V}_{\beta^*} = (\hat{\mathbf{B}}^\top)^{-1}\hat{\boldsymbol{\Sigma}}_{\beta^*}\hat{\mathbf{B}}^{-1}$}\\[2pt]
{\footnotesize $\mathbf{c}^\top\hat{\bfbeta}_{\mathrm{corr}} \pm z_{\alpha/2}\sqrt{\mathbf{c}^\top\mathbf{V}_{\beta^*}\mathbf{c}}$}
};

\node[cibox] at (11.25, -10) (propci) {%
\textsf{\textbf{Propagated CI (\S\ref{sec:bias_ci}) }}\\[3pt]
{\footnotesize $\widehat{\Var}(\hat{\bfbeta}_{\mathrm{corr}}) \approx \mathbf{V}_{\beta^*} + \mathbf{V}_{B}$}\\[2pt]
{\footnotesize $\mathbf{c}^\top\hat{\bfbeta}_{\mathrm{corr}} \pm z_{\alpha/2}\sqrt{\mathbf{c}^\top(\mathbf{V}_{\beta^*} + \mathbf{V}_{B})\mathbf{c}}$}
};

\draw[arr] ([xshift=-1.5cm]corr.south) -- (plugci.north);
\draw[arr] ([xshift=1.5cm]corr.south) -- (propci.north);

\draw[darr] (summOpt.south) -- ++(0,-0.5) -| ([xshift=-1.5cm]propci.north);

\begin{scope}[on background layer]
\node[draw=blue!30, fill=blue!3, rounded corners=5pt, fit=(vlbl)(vdata)(compB), inner sep=10pt] {};
\node[draw=blue!30, fill=blue!3, rounded corners=5pt, fit=(rlbl)(rdata)(compCox), inner sep=10pt] {};
\node[draw=yellow!50, fill=yellow!5, rounded corners=5pt, fit=(slbl)(summ1)(summOpt), inner sep=10pt] {};
\end{scope}

\end{tikzpicture}%
}
\caption{Workflow for bias-corrected inference from AI-extracted covariates in a Cox proportional hazards analysis.}
\label{fig:workflow}
\end{figure}

\subsection{Bias from Covariate Extraction Error}
\label{sec:bias_cov}

We treat continuous and binary covariates within a single multivariate regression calibration framework. This unified approach handles all cross-covariate correlations automatically and requires no assumption of independent errors across covariates.

The \emph{multivariate calibration function}
\begin{equation*}
    \mathbf{m}(\mathbf{x}^*) = \E[\bfX \mid \bfXs = \mathbf{x}^*]
\end{equation*}
maps the full vector of extracted values to the conditional expectation of the truth. Under (A3), $\mathbf{m}(\mathbf{x}^*) = \mathbf{a} + \mathbf{B}\mathbf{x}^*$ with $\mathbf{B}$ and $\mathbf{a}$ as defined in the assumption.

\paragraph{Vendor computation of $\hat{\mathbf{B}}$.} The vendor has $n_v$ paired observations $(\bfX_i, \bfXs_i)$ from the validation dataset. 
The vendor computes $\hat{\mathbf{B}}$ by multivariate OLS: regressing each true covariate column $X_j$ on the full matrix of extracted covariates $\bfXs$, for all $j$. 
In matrix form, let $\mathbb{X} \in \R^{n_v \times p}$ be the matrix of true covariates and $\mathbb{X}^* \in \R^{n_v \times p}$ the matrix of extracted covariates (both centered). 
Then
\begin{equation*}
    \hat{\mathbf{B}} = \bigl(\mathbb{X}^\top\mathbb{X}^*\bigr)\bigl(\mathbb{X}^{*\top}\mathbb{X}^*\bigr)^{-1}.
\end{equation*}
This is equivalent to running $p$ linear regressions with each column of $\mathbb{X}$ as an outcome on all columns of $\mathbb{X}^*$ as covariates, and forming a $p \times p$ coefficient matrix. 
Off-diagonal entries capture how extraction error in one variable predicts the truth of another (e.g., if the AI systematically overestimates age for treated patients, $\hat{\mathbf{B}}_{\mathrm{age,trt}} \ne 0$).

Theorem~\ref{thm:cov_bias} provides a bias derivation that allows us to use $\hat{\mathbf{B}}$ in a corrected estimator:

\begin{theorem}[Covariate extraction bias via calibration]
\label{thm:cov_bias}
Under (A1)--(A3), the naive Cox estimator's bias satisfies
\begin{equation}
    \mathbf{b} \;=\; \bfbeta - \bfbeta^* \;=\; \bigl[(\mathbf{B}^\top)^{-1} - \mathbf{I}\bigr]\,\bfbeta^* \;-\; \boldsymbol{\Omega}_{\mathrm{RC}}^{-1}\,\mathbf{u}_{\mathrm{RC}}(\bfbeta),
    \label{eq:bias_unified}
\end{equation}
where $\boldsymbol{\Omega}_{\mathrm{RC}}^{-1}\,\mathbf{u}_{\mathrm{RC}}(\bfbeta)$ is the \emph{RC residual}: $\mathbf{u}_{\mathrm{RC}}(\bfbeta)$ is the population regression-calibration score evaluated at the true $\bfbeta$, and $\boldsymbol{\Omega}_{\mathrm{RC}}$ is the RC information matrix. This residual is of order
\begin{equation*}
    \|\boldsymbol{\Omega}_{\mathrm{RC}}^{-1}\,\mathbf{u}_{\mathrm{RC}}(\bfbeta)\| \;=\; O(\|\bfbeta\| \cdot \|\E[\boldsymbol{\Sigma}_{X|X^*}]\|^{1/2}),
\end{equation*}
where
\begin{equation*}
    \boldsymbol{\Sigma}_{X|X^*} = \Var(\bfX \mid \bfXs) = \E[(\bfX - \mathbf{m}(\bfXs))(\bfX - \mathbf{m}(\bfXs))^\top \mid \bfXs]
\end{equation*}
is the \emph{calibration residual variance}---the conditional variance of the true covariates given the extracted covariates.
\end{theorem}
Theorem~\ref{thm:cov_bias} is a leading-order asymptotic correction valid in the accurate-extraction, approximately-linear regime; it is not an exact finite-sample result.
It has been shown that regression calibration is not consistent for the Cox model \citep{carroll2006,xie2001,wangsong2021}. 
Bias arises because the induced hazard given calibrated covariates includes a factor depending on $\Var(\bfX \mid \bfXs, T \ge t)$, which varies with $t$ as the risk set composition changes. 
$\boldsymbol{\Sigma}_{X|X^*}$ measures how much uncertainty about $\bfX$ remains after observing $\bfXs$. 
When extraction is perfect up to a linear transformation ($\bfX = \mathbf{a} + \mathbf{B}\bfXs$), $\boldsymbol{\Sigma}_{X|X^*} = \mathbf{0}$ and the RC residual vanishes. 
 
The RC residual is not directly estimable from the observed data, as it depends on the unknown baseline hazard and true $\bfbeta$. 
When the residual is small, which holds when $\|\E[\boldsymbol{\Sigma}_{X|X^*}]\|$ is small (accurate extraction), the bias is dominated by the calibration term:
\begin{equation}
    \mathbf{b} \;\approx\; \bigl[(\mathbf{B}^\top)^{-1} - \mathbf{I}\bigr]\,\bfbeta^*, \qquad \hat{\bfbeta}_{\mathrm{corr}} \;=\; (\hat{\mathbf{B}}^\top)^{-1}\hat{\bfbeta}^*.
    \label{eq:bias_leading}
\end{equation}
The correction is fully computable from the vendor-reported $\hat{\mathbf{B}}$ and the researcher's naive estimate $\hat{\bfbeta}^*$. 
When the RC residual is suspected to be non-negligible, its potential magnitude should be assessed with a sensitivity diagnostic (Section~\ref{sec:sensitivity}).

Table~\ref{tab:stats} consolidates the validation summaries the vendor must report to make this correction and the inference that follows available to a downstream researcher: the calibration slope matrix $\hat{\mathbf{B}}$ alone suffices for the point correction, while the remaining entries support the propagated intervals of Section~\ref{sec:bias_ci} and the sensitivity diagnostic of Section~\ref{sec:sensitivity}.

\begin{table}[ht]
\centering
\caption{Vendor-reported validation summary statistics.}
\label{tab:stats}
\begin{tabular}{@{}lll@{}}
\toprule
\textbf{Purpose} & \textbf{Summary statistic} & \textbf{Symbol} \\
\midrule
Point correction & Full calibration slope matrix & $\hat{\mathbf{B}}$, $p \times p$ \\
\addlinespace
Propagated inference & Residual covariance of calibration regressions & $\hat{\boldsymbol{\Sigma}}_{\mathrm{resid}}$, $p \times p$ \\
 & Design Gram inverse & $(\mathbb{X}^{*\top}\mathbb{X}^*)^{-1}$, $p \times p$ \\
\addlinespace
Sensitivity diagnostic & Calibration residual variance norm & $ \|\hat{\boldsymbol{\Sigma}}_{X|X^*}\|_{\mathrm{op}}^{1/2} $  \\
\bottomrule
\end{tabular}
\end{table}

\paragraph{The binary covariate setting.}
While Theorem~\ref{thm:cov_bias} covers continuous, binary, and mixed covariates, the discrete case is where the leading-order correction is most strained.
The linear calibration assumption (A3) is only an approximation for a binary covariate: the true calibration function $\E[X_j \mid \bfXs] = \Prob(X_j = 1 \mid \bfXs)$ is a conditional probability confined to $[0,1]$, whereas the linear form $\mathbf{a} + \mathbf{B}\mathbf{x}^*$ is unbounded.
The vendor's OLS fit therefore recovers the best \emph{linear} predictor of $X_j$ rather than the conditional mean itself, leaving a curvature error $\Var \big( \E[X_j\mid\bfXs] - (\mathbf{a}+\mathbf{B}\bfXs)_j \big)$ that is zero under Gaussian (A3) but generally nonzero here.
A fully nonparametric calibration of $\E[\bfX \mid \bfXs]$ would eliminate it, at the cost of requiring the vendor to report richer, nonlinear calibration summaries.

\subsection{Bias-Adjusted Confidence Intervals}
\label{sec:bias_ci}

The bias formula in Theorem~\ref{thm:cov_bias} expresses the difference $\bfbeta - \bfbeta^*$ in terms of estimable quantities and provides a corrected estimator $\hat{\bfbeta}_{\mathrm{corr}}$. 
While this is valid as a point estimate, the extraction error also needs to be incorporated into statistical inference.

Because the calibration matrix $\mathbf{B}$ is generally non-diagonal with off-diagonal entries encoding cross-covariate extraction error correlations, inference should be conducted at the vector level rather than coordinate-wise. 
To leading order in the RC residual, Theorem~\ref{thm:cov_bias} gives
\begin{equation*}
    \hat{\bfbeta}^* \;=\; \mathbf{B}^\top\bfbeta \;+\; \boldsymbol{\varepsilon}, \qquad \boldsymbol{\varepsilon} \sim N(\mathbf{0}, \boldsymbol{\Sigma}_{\beta^*}),
\end{equation*}
where $\boldsymbol{\Sigma}_{\beta^*}$ is the sampling covariance of $\hat{\bfbeta}^*$. The vector-corrected estimator is
\begin{equation*}
    \hat{\bfbeta}_{\mathrm{corr}} \;=\; (\hat{\mathbf{B}}^\top)^{-1}\hat{\bfbeta}^*.
\end{equation*}
Wald inference on any linear contrast $\mathbf{c}^\top\bfbeta$ (e.g., $\mathbf{c} = \mathbf{e}_k$ for the $k$th coefficient) proceeds from the multivariate delta method applied to this correction.

\subsubsection{Plug-in CI (calibration treated as known)}

Assuming $\hat{\mathbf{B}} \approx \mathbf{B}$ with negligible estimation error, the Jacobian is $\partial\hat{\bfbeta}_{\mathrm{corr}}/\partial\hat{\bfbeta}^* = (\hat{\mathbf{B}}^\top)^{-1}$, and the covariance of the corrected estimator is
\begin{equation}\label{eq:V_beta_star}
    \mathbf{V}_{\beta^*} \;=\; (\hat{\mathbf{B}}^\top)^{-1}\,\hat{\boldsymbol{\Sigma}}_{\beta^*}\,\hat{\mathbf{B}}^{-1}.
\end{equation}
For a contrast $\mathbf{c}^\top\bfbeta$, the plug-in CI is
\begin{equation*}
    \mathrm{CI}_{\mathrm{plug}}(\mathbf{c}) \;=\; \mathbf{c}^\top\hat{\bfbeta}_{\mathrm{corr}} \;\pm\; z_{\alpha/2}\sqrt{\mathbf{c}^\top\mathbf{V}_{\beta^*}\mathbf{c}}.
\end{equation*}
The marginal CI for $\beta_k$ is obtained with $\mathbf{c} = \mathbf{e}_k$ and uses the $k$th diagonal entry of $\mathbf{V}_{\beta^*}$. This approach is appropriate when the validation sample $n_v$ is large relative to the study, so sampling variability in $\hat{\mathbf{B}}$ is negligible compared to that in $\hat{\bfbeta}^*$.

\subsubsection{Propagated CI (calibration estimated with uncertainty)}

When $\hat{\mathbf{B}}$ carries non-negligible estimation error, the variance must also reflect uncertainty in the vendor's calibration fit. 
Treating $\hat{\bfbeta}^*$ and $\hat{\mathbf{B}}$ as independent (the study fit is independent of the validation fit), the multivariate delta method applied to $\hat{\bfbeta}_{\mathrm{corr}} = (\hat{\mathbf{B}}^\top)^{-1}\hat{\bfbeta}^*$ yields two additive variance contributions $\widehat{\Var}(\hat{\bfbeta}_{\mathrm{corr}}) \;\approx\; \mathbf{V}_{\beta^*} \;+\; \mathbf{V}_{B}.$

The study-data component $\mathbf{V}_{\beta^*}$ is \eqref{eq:V_beta_star}.
Assuming all $p$ calibration regressions share the same extracted-covariate design $\mathbb{X}^*$, the covariance of $\hat{\mathbf{B}}$ factors as $\widehat{\Cov}(\hat{B}_{ij}, \hat{B}_{i'j'}) = [\hat{\boldsymbol{\Sigma}}_{\mathrm{resid}}]_{ii'}\,[(\mathbb{X}^{*\top}\mathbb{X}^*)^{-1}]_{jj'}$, where $\hat{\boldsymbol{\Sigma}}_{\mathrm{resid}}$ is the $p\times p$ residual covariance across the $p$ calibration regressions, computed as follows.
Consider the $n_v \times p$ residual matrix
\begin{equation*}
    \mathbf{Z} \;=\;
    \begin{pmatrix}
        \mathbf{z}_1^\top \\ \vdots \\ \mathbf{z}_{n_v}^\top
    \end{pmatrix}
    \;=\; \mathbb{X} \;-\; \mathbb{X}^*\hat{\mathbf{B}}^\top,
\end{equation*}
whose $(i,j)$ entry is the residual of subject $i$ in the $j$-th calibration regression. The residual covariance matrix is the sample cross-product, scaled by the residual degrees of freedom:
\begin{equation*}
    \hat{\boldsymbol{\Sigma}}_{\mathrm{resid}} \;=\; \frac{1}{n_v - p - 1}\,\mathbf{Z}^\top\mathbf{Z} \;\in\; \R^{p \times p}.
\end{equation*}

Propagating this through $\partial\hat{\bfbeta}_{\mathrm{corr}}/\partial\hat{B}_{ij} = -[\hat{\bfbeta}_{\mathrm{corr}}]_i\,(\hat{\mathbf{B}}^\top)^{-1}\mathbf{e}_j$ (from $\partial M^{-1}/\partial M_{ab} = -M^{-1}\mathbf{e}_a\mathbf{e}_b^\top M^{-1}$) yields
\begin{equation*}
    \mathbf{V}_{B} \;=\; \bigl(\hat{\bfbeta}_{\mathrm{corr}}^\top\,\hat{\boldsymbol{\Sigma}}_{\mathrm{resid}}\,\hat{\bfbeta}_{\mathrm{corr}}\bigr)\,(\hat{\mathbf{B}}^\top)^{-1}\,(\mathbb{X}^{*\top}\mathbb{X}^*)^{-1}\,\hat{\mathbf{B}}^{-1}.
\end{equation*}
This yields the Wald confidence interval
\begin{align*}
    \mathrm{CI}_{\mathrm{prop}}(\mathbf{c}) &\;=\; \mathbf{c}^\top\hat{\bfbeta}_{\mathrm{corr}} \;\pm\; z_{\alpha/2}\sqrt{\mathbf{c}^\top\bigl(\mathbf{V}_{\beta^*} + \mathbf{V}_{B}\bigr)\mathbf{c}}.
\end{align*}
The vendor reports two $p \times p$ matrices for propagated inference: the residual covariance $\hat{\boldsymbol{\Sigma}}_{\mathrm{resid}}$ and the design Gram inverse $(\mathbb{X}^{*\top}\mathbb{X}^*)^{-1}$ from the calibration regressions. When $\hat{\mathbf{B}}$ is approximately diagonal, the formulas collapse to coordinate-wise scalar expressions:
\begin{equation*}
    \Var(\hat{\beta}_{k,\mathrm{corr}}) \;\approx\; \frac{\hat{\sigma}_{\beta_k^*}^2}{\hat{B}_{kk}^2} \;+\; \frac{\hat{\beta}_{k,\mathrm{corr}}^2\,\widehat{\Var}(\hat{B}_{kk})}{\hat{B}_{kk}^2};
\end{equation*}
the general vector form should be used whenever off-diagonal entries of $\hat{\mathbf{B}}$ are non-negligible, representing cross-variable correlation of extraction errors.

If the vendor fits the $p$ calibration regressions with column-specific designs (due to e.g., variable selection or regularization applied separately per outcome), the Kronecker factorization above does not hold and $\mathbf{V}_B$ takes the general form $(\hat{\mathbf{B}}^\top)^{-1}\mathbf{Q}(\hat{\bfbeta}_{\mathrm{corr}})\hat{\mathbf{B}}^{-1}$ with $[\mathbf{Q}(\mathbf{v})]_{jj'} = \sum_{i,i'} v_i v_{i'}\,\widehat{\Cov}(\hat{B}_{ij}, \hat{B}_{i'j'})$, requiring the vendor to report the full covariance of $\hat{\mathbf{B}}$.
Note also that regularization biases $\hat{\mathbf{B}}$ toward zero; this bias propagates into $\hat{\bfbeta}_{\mathrm{corr}}$ and should be reported alongside the covariance so the researcher can account for it.

\subsubsection{Sensitivity analysis for the RC residual}
\label{sec:sensitivity}

The plug-in and propagated CIs account for sampling variability in $\hat{\bfbeta}^*$ and (in the propagated case) estimation uncertainty in $\hat{\mathbf{B}}$, but both treat the leading-order approximation $\bfbeta \approx (\mathbf{B}^\top)^{-1}\bfbeta^*$ as exact.
Theorem~\ref{thm:cov_bias} shows that the neglected RC residual contributes an additional systematic error of order $O(\|\bfbeta\|\cdot\|\E[\boldsymbol{\Sigma}_{X|X^*}]\|^{1/2})$.
This term cannot be estimated from the observed data, but it can be bounded, and the bound provides a sensitivity diagnostic for judging whether the leading-order correction is adequate.

Suppose the vendor reports the estimated calibration residual variance $\hat{\boldsymbol{\Sigma}}_{\mathrm{resid}}$, which is equal to $\boldsymbol{\Sigma}_{X|X^*}$ in expectation under (A3). 
Let $\bar\sigma = \|\hat{\boldsymbol{\Sigma}}_{\mathrm{resid}}\|_{\mathrm{op}}^{1/2}$ denote the square root of its largest eigenvalue.
The RC residual bound implies that the true $\bfbeta$ lies within
\begin{equation*}
    \|\bfbeta - \hat{\bfbeta}_{\mathrm{corr}}\| \;\le\; \underbrace{\text{(sampling error)}}_{\text{covered by Plug-in/Propagated CI}} \;+\; \underbrace{C\,\|\hat{\bfbeta}_{\mathrm{corr}}\|\,\bar\sigma}_{\text{RC residual bound}},
\end{equation*}
where $C$ is an unknown constant depending on the covariate, hazard, and censoring distributions.

To judge whether the RC residual would materially change the inference, define the ratio of the RC residual bound to the propagated CI half-width:
\begin{equation}
    \rho(\mathbf{c}) \;=\; \frac{|\mathbf{c}^\top\hat{\bfbeta}_{\mathrm{corr}}|\,\bar\sigma}{z_{\alpha/2}\sqrt{\mathbf{c}^\top(\mathbf{V}_{\beta^*} + \mathbf{V}_B)\mathbf{c}}}.
    \label{eq:rho}
\end{equation}
When $\rho \ll 1$, the RC residual is small relative to the statistical uncertainty already reflected in the CI, and the leading-order correction is adequate.
When $\rho$ is appreciable, the propagated CI may understate the true uncertainty, and should be interpreted with caution.
Because $\rho$ omits the unknown constant $C$, it screens for the \emph{potential} of excess bias rather than certifying its absence; a small $\rho$ indicates the residual is dominated by statistical uncertainty for any plausible $C$, not that the residual is zero.

\section{Simulation studies}
\label{sec:sims}

In this section, we implement simulation studies to examine the performance of our bias-corrected estimator and confidence intervals.
We simulate $n = n_{\text{study}} + n_v$ subjects with $p = 4$ covariates subject to AI-extraction error.
The true covariate vector $X = (X_1, X_2, X_3, X_4)^\top$ is generated with $X_1 \sim N(0, 1)$ and $X_2 \sim N(0, 1)$ as continuous covariates, and $X_3 \sim \text{Bernoulli}(0.5)$ and $X_4 \sim \text{Bernoulli}(0.5)$ as binary covariates.  
All four are mutually independent.
The event time follows a Cox model with exponential baseline hazard,
\[
  \lambda(t \mid X) = \lambda_0 \exp(\bfbeta^\top X),
\]
Censoring times are independently generated as $C \sim \text{Exponential}(\nu)$.
The observed data are $(Y, \Delta) = (\min(T, C),\; \ind\{T \le C\})$.
In both studies, the true coefficient vector is $\bfbeta = (0.60,\; -0.40,\; 0.50,\; -0.35)^\top$.
The baseline hazard rate is $\lambda_0 = 0.1$ and the censoring rate is $\nu = 0.05$.
This results in an observed event rate of 0.66.
$n_v = 300$ subjects form the validation set, in which the vendor has access to paired observations $(\bfX_i, \bfXs_i)$.  
The remaining $n_{\text{study}} = 1500$ subjects form the study set, in which the researcher observes only $(\bfXs_i, Y_i, \Delta_i)$ and the vendor-reported calibration summaries.

In each simulation replicate, the vendor computes $\hat{\mathbf{B}}$, $\hat{\boldsymbol{\Sigma}}_{\mathrm{resid}}$, and $(\mathbb{X}^{*\top}\mathbb{X}^*)^{-1}$ from the validation data.  
The researcher then fits the naive Cox model $\hat{\bfbeta}^*$ on the study data using $\bfXs$ in place of $\bfX$, and computes the corrected estimate $\hat{\bfbeta}_{\mathrm{corr}} = (\hat{\mathbf{B}}^\top)^{-1} \hat{\bfbeta}^*$ together with plug-in and propagated confidence intervals.
Each scenario is replicated $R = 500$ times.  
For each covariate $k$ and estimator (naive or corrected), we report the signed bias $R^{-1} \sum_{r=1}^{R} (\hat{\beta}^{(r)}_k - \beta_k)$ and the root mean squared error $\bigl[R^{-1} \sum_{r=1}^{R} (\hat{\beta}^{(r)}_k - \beta_k)^2\bigr]^{1/2}$.  
We also report the CI coverage (the proportion of replicates in which the 95\% confidence interval contains the true $\beta_k$) for each CI method. 
All simulations are run in R version 4.4.1 \citep{r-cite}, using the \texttt{survival} package version 3.8-3 \citep{survival-package} for fitting Cox models with the default Efron approximation for ties.

We additionally conduct a study of real data with synthetic covariate measurement error applied.
This is reported in Appendix~\ref{app:mimic}.

\subsection{Study 1: Cross-Dependent Extraction Errors}
\label{sec:cross-dependent}
In this study, a shared latent factor $Z_i \sim N(0,1)$, drawn independently per subject and independently of $X$, drives correlated extraction errors across all four covariates.
The extracted continuous covariates ($j = 1, 2$) are generated as
\[
  X_j^* = X_j + \lambda_j Z + \varepsilon_j, \qquad \varepsilon_j \stackrel{\text{iid}}{\sim} N(0, \sigma_{\varepsilon,j}^2),
\]
where $\varepsilon_j, Z, X$ are mutually independent.  
The factor loading $\lambda_j$ induces cross-covariate correlation through the shared~$Z$, while $\varepsilon_j$ adds covariate-specific error.
For $X_1$, the factor loading is $\lambda_1 = 0.6$ with noise $\sigma_{\varepsilon,1} = 0.25$; for $X_2$, $\lambda_2 = 0.5$ and $\sigma_{\varepsilon,2} = 0.20$.
For binary covariates ($j = 3, 4$), the misclassification probability depends on the same latent factor:
\[
  P(X_j^* \ne X_j \mid Z) = \text{logit}^{-1}(\alpha_j + \gamma_j |Z|).
\]
For $X_3$, the intercept is $\alpha_3 = -2.5$ and the sensitivity to the latent factor is $\gamma_3 = 1.5$; for $X_4$, $\alpha_4 = -2.8$ and $\gamma_4 = 1.2$.
Because both the continuous additive error ($\lambda_j Z$) and the binary flip probability ($\gamma_j |Z|$) depend on the same~$Z$, extraction errors are correlated across all four covariates.  
The multivariate calibration matrix $\mathbf{B}$ captures this structure through its off-diagonal entries.

Table~\ref{tab:cross-dependent} summarizes the simulation results.
Extraction quality is moderate to high across all covariates.
The naive estimator exhibits substantial bias: $X_1$ is attenuated by $-0.114$ (19\% of $\beta_1 = 0.60$) and $X_3$ by $-0.256$ (50\% of $\beta_3 = 0.50$), while $X_4$ is biased by $+0.121$ (35\% of $|\beta_4| = 0.35$).
The correction eliminates most of this bias, reducing signed bias to below 0.03 in absolute value for all covariates.
Notably, $X_2$ has near-zero naive bias; the correction slightly increases its absolute bias to $0.02$, reflecting the variance cost of the matrix inversion in $(\hat{\mathbf{B}}^\top)^{-1}$.
RMSE follows the same pattern: the correction reduces RMSE for most variables where bias dominates the naive error, but slightly inflates RMSE for $X_2$, where the naive estimator was already nearly unbiased and the correction adds estimation variance.

\begin{table}[htbp]
\centering
\caption{
Simulation results for cross-dependent extraction errors (Study~1).
Each cell reports the point estimate with a 95\% Monte Carlo confidence interval in brackets.
Extraction quality is $R^2$ for continuous covariates ($X_1$, $X_2$) and classification accuracy for binary covariates ($X_3$, $X_4$).}
\label{tab:cross-dependent}
\smallskip
\resizebox{\textwidth}{!}{%
\begin{tabular}{l c cc cc ccc}
\toprule
& & \multicolumn{2}{c}{Signed Bias} & \multicolumn{2}{c}{RMSE} & \multicolumn{3}{c}{Coverage (\%)} \\
\cmidrule(lr){3-4} \cmidrule(lr){5-6} \cmidrule(lr){7-9}
Covariate & Ext.\ Quality & Naive & Corrected & Naive & Corrected & Naive & Plug-in & Propagated \\
\midrule
$X_1$
  & 0.70 & $-$0.114 & $-$0.029
  & 0.118 & 0.057
  & 3.6 & 82.4 & 98.4 \\
  & {\scriptsize[0.70, 0.70]} & {\scriptsize[$-$0.117, $-$0.111]} & {\scriptsize[$-$0.033, $-$0.024]}
  & {\scriptsize[0.115, 0.121]} & {\scriptsize[0.053, 0.060]}
  & {\scriptsize[2.0, 5.2]} & {\scriptsize[79.1, 85.7]} & {\scriptsize[97.3, 99.5]} \\[4pt]
$X_2$
  & 0.77 & $-$0.008 & \phantom{$-$}0.020
  & 0.031 & 0.048
  & 94.4 & 88.0 & 99.2 \\
  & {\scriptsize[0.77, 0.78]} & {\scriptsize[$-$0.010, $-$0.005]} & {\scriptsize[0.017, 0.024]}
  & {\scriptsize[0.029, 0.033]} & {\scriptsize[0.045, 0.051]}
  & {\scriptsize[92.4, 96.4]} & {\scriptsize[85.2, 90.8]} & {\scriptsize[98.4, 100.0]} \\[4pt]
$X_3$
  & 0.76 & $-$0.256 & $-$0.015
  & 0.264 & 0.148
  & 3.8 & 90.0 & 94.4 \\
  & {\scriptsize[0.75, 0.76]} & {\scriptsize[$-$0.262, $-$0.251]} & {\scriptsize[$-$0.028, $-$0.002]}
  & {\scriptsize[0.259, 0.270]} & {\scriptsize[0.139, 0.158]}
  & {\scriptsize[2.1, 5.5]} & {\scriptsize[87.4, 92.6]} & {\scriptsize[92.4, 96.4]} \\[4pt]
$X_4$
  & 0.84 & \phantom{$-$}0.121 & \phantom{$-$}0.012
  & 0.136 & 0.108
  & 50.0 & 92.2 & 95.2 \\
  & {\scriptsize[0.84, 0.84]} & {\scriptsize[0.115, 0.126]} & {\scriptsize[0.003, 0.022]}
  & {\scriptsize[0.131, 0.141]} & {\scriptsize[0.102, 0.114]}
  & {\scriptsize[45.6, 54.4]} & {\scriptsize[89.8, 94.6]} & {\scriptsize[93.3, 97.1]} \\
\bottomrule
\end{tabular}}
\end{table}

Naive confidence intervals have poor coverage in general.
Plug-in CIs, which treat $\hat{\mathbf{B}}$ as known, improve coverage substantially (82--92\%) but remain below the nominal 95\% level for most covariates, indicating that calibration uncertainty is non-negligible with $n_v = 300$.
Propagated CIs, which account for uncertainty in $\hat{\mathbf{B}}$, achieve near or above nominal coverage for all covariates.
Appendix~\ref{app:prop-ci-cvg} demonstrates the convergence of CI coverage with increasing $n_v$ under this study's data-generating process.

\subsection{Study 2: Nonlinear Calibration (Constant Total Error)}
\label{sec:nonlinear}

This study tests robustness to deliberate, controlled violations of assumption (A3).
The extracted continuous covariates ($j = 1, 2$) are generated as
\[
X_j^* = X_j + \underbrace{\delta_j X_j^2}_{\text{quadratic distortion}} + \underbrace{\sigma_{0,j}(1 + \kappa_j X_j^2)\,\varepsilon_j}_{\text{heteroscedastic noise}}, \qquad \varepsilon_j \sim N(0, 1).
\]
The error has two nonlinear components.
The quadratic distortion term $\delta_j X_j^2$ introduces a systematic, mean-shifting nonlinearity.
The heteroscedastic noise term, governed by $\kappa_j$, makes the extraction noisier for extreme covariate values.
Together, $\text{MSE}_j = 3\delta_j^2 + \sigma_{0,j}^2 (1 + 2\kappa_j + 3\kappa_j^2)$.
For binary covariates ($j = 3, 4$), the misclassification probability is modulated by the true continuous covariates through
\[
P(X_j^* \ne X_j \mid X_1, X_2) = p_j \bigl(1 + \eta_j \tanh(X_1 + X_2)\bigr),
\]
where $p_j$ is the target marginal flip rate and $\eta_j \in [0, 1)$ controls the strength of the covariate dependence.
Since $\tanh$ is an odd function and $X_1 + X_2$ is symmetric around zero, $E[\tanh(X_1 + X_2)] = 0$, so the marginal flip rate is exactly $p_j$ regardless of $\eta_j$.
When $\eta_j = 0$, flips are uniform and the calibration $E[X_j \mid \bfXs]$ is linear.
When $\eta_j > 0$, the flip probability varies smoothly from $p_j(1 - \eta_j)$ to $p_j(1 + \eta_j)$ across the full range of $X_1 + X_2$, making $E[X_j \mid \bfXs]$ a nonlinear function of $(X_1^*, X_2^*)$ and violating (A3).

The continuous errors $\varepsilon_1$ and $\varepsilon_2$ are independent, and the binary flips for $X_3$ and $X_4$ are conditionally independent given $(X_1, X_2)$.
The only source of dependence between extracted covariates is through the true covariates themselves: $X_3^*$ and $X_4^*$ both depend on $X_1 + X_2$ (which governs their flip probabilities), so they are marginally correlated with $X_1^*$ and $X_2^*$.
Off-diagonal entries in the estimated $\hat{\mathbf{B}}$ will be nonzero, but the extraction \emph{errors} are conditionally uncorrelated given the true covariates.
This contrasts with Study~1, where the shared $Z$ induces genuine cross-covariate error correlation.
We investigate three levels of nonlinearity:
\begin{enumerate}
\item Under \textbf{mild nonlinearity}, $\delta = 0.02$, $\kappa = 0.05$, and $\eta = 0.10$.
\item Under \textbf{moderate nonlinearity}, $\delta = 0.15$, $\kappa = 0.30$, and $\eta = 0.50$.
\item Under \textbf{severe nonlinearity}, $\delta = 0.25$, $\kappa = 0.55$, and $\eta = 0.90$.
\end{enumerate}
To isolate the effect of nonlinearity shape, $\sigma_{0,j}$ is calibrated so that the marginal MSE is identical across severity levels, with targets $\text{MSE}_1 = \text{MSE}_2 = 0.70$.
The binary flip rates are fixed at $p_3 = p_4 = 0.20$ and are automatically preserved by the $\tanh$ modulation.
Thus, any change in bias, RMSE, or CI coverage between severity levels is attributable to the nonlinearity of the calibration function, not to the amount of error.

Table~\ref{tab:nonlinear} displays RMSE and CI coverage as nonlinearity increases.
Naive RMSE is large and mostly flat across severity levels for all four covariates, consistent with the constant-error design: the naive estimator is always biased, and the shape of the calibration function does not affect its magnitude.
Corrected RMSE is substantially smaller at all severity levels, while increasing with nonlinearity for covariates $X_1, X_4$.

\begin{table}[ht]
\centering
\setlength{\tabcolsep}{4pt}
\caption{Simulation results under nonlinear calibration (Study 2). Each cell gives the point estimate over its 95\% Monte Carlo CI.}
\label{tab:nonlinear}
\begin{tabular}{llrrrrr}
\toprule
& & \multicolumn{2}{c}{RMSE} & \multicolumn{3}{c}{95\% CI Coverage (\%)} \\
\cmidrule(lr){3-4} \cmidrule(lr){5-7}
Cov. & Severity & Naive & Corr. & Naive & Plug-in & Prop. \\
\midrule
\multirow{3}{*}{$X_1$} & Mild & \shortstack{0.290 \\ {\scriptsize [0.288, 0.292]}} & \shortstack{0.090 \\ {\scriptsize [0.086, 0.094]}} & \shortstack{0.0 \\ {\scriptsize [0.0, 0.0]}} & \shortstack{57.8 \\ {\scriptsize [53.5, 62.1]}} & \shortstack{95.6 \\ {\scriptsize [93.8, 97.4]}} \\
 & Moderate & \shortstack{0.288 \\ {\scriptsize [0.286, 0.290]}} & \shortstack{0.094 \\ {\scriptsize [0.090, 0.098]}} & \shortstack{0.0 \\ {\scriptsize [0.0, 0.0]}} & \shortstack{53.4 \\ {\scriptsize [49.0, 57.8]}} & \shortstack{91.6 \\ {\scriptsize [89.2, 94.0]}} \\
 & Severe & \shortstack{0.293 \\ {\scriptsize [0.290, 0.296]}} & \shortstack{0.114 \\ {\scriptsize [0.108, 0.119]}} & \shortstack{0.0 \\ {\scriptsize [0.0, 0.0]}} & \shortstack{35.8 \\ {\scriptsize [31.6, 40.0]}} & \shortstack{81.2 \\ {\scriptsize [77.8, 84.6]}} \\
\midrule
\multirow{3}{*}{$X_2$} & Mild & \shortstack{0.194 \\ {\scriptsize [0.192, 0.196]}} & \shortstack{0.074 \\ {\scriptsize [0.070, 0.078]}} & \shortstack{0.0 \\ {\scriptsize [0.0, 0.0]}} & \shortstack{71.4 \\ {\scriptsize [67.4, 75.4]}} & \shortstack{96.8 \\ {\scriptsize [95.3, 98.3]}} \\
 & Moderate & \shortstack{0.186 \\ {\scriptsize [0.184, 0.188]}} & \shortstack{0.066 \\ {\scriptsize [0.062, 0.070]}} & \shortstack{0.0 \\ {\scriptsize [0.0, 0.0]}} & \shortstack{80.4 \\ {\scriptsize [76.9, 83.9]}} & \shortstack{98.2 \\ {\scriptsize [97.0, 99.4]}} \\
 & Severe & \shortstack{0.176 \\ {\scriptsize [0.173, 0.178]}} & \shortstack{0.069 \\ {\scriptsize [0.065, 0.073]}} & \shortstack{0.0 \\ {\scriptsize [0.0, 0.0]}} & \shortstack{77.6 \\ {\scriptsize [73.9, 81.3]}} & \shortstack{97.2 \\ {\scriptsize [95.8, 98.6]}} \\
\midrule
\multirow{3}{*}{$X_3$} & Mild & \shortstack{0.239 \\ {\scriptsize [0.233, 0.245]}} & \shortstack{0.151 \\ {\scriptsize [0.142, 0.160]}} & \shortstack{5.8 \\ {\scriptsize [3.8, 7.8]}} & \shortstack{84.2 \\ {\scriptsize [81.0, 87.4]}} & \shortstack{88.0 \\ {\scriptsize [85.2, 90.8]}} \\
 & Moderate & \shortstack{0.246 \\ {\scriptsize [0.240, 0.252]}} & \shortstack{0.148 \\ {\scriptsize [0.139, 0.157]}} & \shortstack{5.0 \\ {\scriptsize [3.1, 6.9]}} & \shortstack{85.4 \\ {\scriptsize [82.3, 88.5]}} & \shortstack{88.0 \\ {\scriptsize [85.2, 90.8]}} \\
 & Severe & \shortstack{0.239 \\ {\scriptsize [0.234, 0.245]}} & \shortstack{0.150 \\ {\scriptsize [0.140, 0.160]}} & \shortstack{5.0 \\ {\scriptsize [3.1, 6.9]}} & \shortstack{83.4 \\ {\scriptsize [80.1, 86.7]}} & \shortstack{86.6 \\ {\scriptsize [83.6, 89.6]}} \\
\midrule
\multirow{3}{*}{$X_4$} & Mild & \shortstack{0.178 \\ {\scriptsize [0.173, 0.184]}} & \shortstack{0.144 \\ {\scriptsize [0.135, 0.152]}} & \shortstack{25.8 \\ {\scriptsize [22.0, 29.6]}} & \shortstack{85.8 \\ {\scriptsize [82.7, 88.9]}} & \shortstack{87.6 \\ {\scriptsize [84.7, 90.5]}} \\
 & Moderate & \shortstack{0.176 \\ {\scriptsize [0.170, 0.182]}} & \shortstack{0.151 \\ {\scriptsize [0.141, 0.160]}} & \shortstack{30.6 \\ {\scriptsize [26.6, 34.6]}} & \shortstack{84.6 \\ {\scriptsize [81.4, 87.8]}} & \shortstack{88.0 \\ {\scriptsize [85.2, 90.8]}} \\
 & Severe & \shortstack{0.176 \\ {\scriptsize [0.171, 0.182]}} & \shortstack{0.158 \\ {\scriptsize [0.147, 0.168]}} & \shortstack{28.4 \\ {\scriptsize [24.4, 32.4]}} & \shortstack{83.4 \\ {\scriptsize [80.1, 86.7]}} & \shortstack{85.6 \\ {\scriptsize [82.5, 88.7]}} \\
\bottomrule
\end{tabular}
\end{table}

Naive coverage is near zero for $X_1$ and $X_2$ at all severity levels, reflecting severe attenuation bias in those covariates.
Plug-in coverage improves substantially over naive but falls short of the nominal level, as it does not account for uncertainty in $\hat{\mathbf{B}}$.
Propagated CI coverage degrades with severity, maintaining nominal coverage only for $X_2$.
The binary covariates ($X_3$, $X_4$) exhibit below-nominal propagated coverage even at mild severity (88.0\% and 87.6\%), reflecting residual bias from the calibration function that the leading-order correction does not fully capture.

\section{Discussion}
\label{sec:discussion}

This paper develops a framework for quantifying and correcting the bias introduced by AI-extracted covariates in the Cox proportional hazards model, under the assumption that event times and indicators are accurately observed.
The central theoretical result (Theorem~\ref{thm:cov_bias}) yields a fully computable corrected estimator $\hat{\bfbeta}_{\mathrm{corr}} = (\hat{\mathbf{B}}^\top)^{-1}\hat{\bfbeta}^*$ that requires only the vendor-reported $\hat{\mathbf{B}}$ and the researcher's naive Cox fit.
Bias-adjusted confidence intervals provide inference that honestly reflects both sampling variability in the study data and estimation uncertainty in the vendor's validation sample.
Our framework cleanly separates vendor and researcher responsibilities: the vendor should report the calibration summaries in Table~\ref{tab:stats} and the researcher combines these with the output of any standard Cox software.
We encourage vendors to adopt Table~\ref{tab:stats} as a reporting standard for any AI-extracted covariate dataset intended for use in survival analysis, particularly as extensions to emerging and existing guidelines for real-world-evidence reporting.

R code for implementing our bias-correction methodology and for replicating the empirical studies described in Section~\ref{sec:sims} and Appendix~\ref{app:mimic} is available at \url{https://github.com/asondhi/cox-bias-cov}.

\paragraph{Limitations.}

The framework rests on several assumptions whose relaxation defines natural directions for future work.
First, we assume that the event time~$T$ and event indicator~$\Delta$ are observed exactly.
This is reasonable when outcomes come from administrative or registry sources that do not require language extraction, but it excludes settings where outcomes are themselves AI-extracted; for example, when the event of interest (disease recurrence, treatment discontinuation) must be identified from unstructured clinical notes.

Each error source enters the partial likelihood differently. 
Covariate error leaves every risk set correct and produces the linear attenuation of Theorem~\ref{thm:cov_bias}, which vanishes at $\bfbeta = \mathbf{0}$.
Misclassified $\Delta$ changes only which records count as events.
Events missed at a rate that does not depend on $\bfX$ or $t$ merely thin the event set and are absorbed into $\lambda_0$, so imperfect sensitivity costs precision but not consistency.
False events do bias $\bfbeta$: a censored record is labeled an event at the time its follow-up ends, not at a time generated by the failure hazard, so its covariate is effectively a draw from the \emph{unweighted} risk set at that time, whereas the true score subtracts the risk-set mean $\bar{\bfX}(\bfbeta, t)$ weighted by $\exp(\bfbeta^\top\bfX)$; the discrepancy is $-\Var(\bfX \mid T \ge t)\bfbeta$ to first order and therefore pulls the estimate toward zero.
Error in $T$ perturbs event ordering and risk-set membership together, so the partial likelihood is invariant to a common monotone distortion of the time scale but is biased by heterogeneous error.

When $T$ or $\Delta$ is subject to extraction error, additional bias terms arise from the interaction between outcome misclassification and the partial likelihood structure, and the non-differential error assumption requires a different form.
Extending the framework to handle errors in $T$, $\Delta$, or both is an important open problem; the joint treatment of covariate and outcome extraction error in the Cox model has not, to our knowledge, been addressed.

Second, the present framework addresses estimation of the full-population coefficient vector $\bfbeta$ and does not cover subpopulation analyses in which the analytic sample is selected on the basis of an AI-extracted variable.
Such analyses are common in practice: for example, a researcher may restrict attention to patients with a particular AI-extracted diagnosis or biomarker status, effectively conditioning on $X_j^*$ falling in a specific category.
Here, even if the extraction error is non-differential in the full population, the subpopulation defined by $X_j^* = 1$ will systematically differ from the subpopulation defined by $X_j = 1$ due to misclassification, and the induced bias depends on both the misclassification rates and the joint distribution of $X_j$ with the other covariates and with the outcome.

Finally, the bias formulas are leading-order approximations that rely on the linear calibration assumption~(A3).
As the simulations in Study~2 show, the correction remains effective under mild departures from linearity but degrades under severe nonlinearity, particularly for binary covariates.
Nonlinear extensions---for example, replacing the linear calibration model with a flexible nonparametric or semiparametric specification of $\E[\bfX \mid \bfXs]$---would broaden the applicability of the framework, though at the cost of requiring the vendor to report richer summaries.

\bibliographystyle{apalike}
\bibliography{references}

\newpage

\appendix
\section{Proof of Theorem 1}
\label{app:proofs}

\paragraph{Asymptotic regime and conventions.} All asymptotic bounds in this appendix are stated in the limit $\|\E[\boldsymbol{\Sigma}_{X|X^*}]\| \to 0$ (the accurate-extraction regime), with $\bfbeta \in \mathcal{B}$ held fixed per (A1). 
Under (A1), hidden constants in big-$O$ bounds depend on $\mathcal{B}$ and the covariate, hazard, and censoring distributions. 
Specifically, under the joint sub-Gaussianity in (A1)(ii), the moment generating function $M_X(\bfbeta) = \E[\exp(\bfbeta^\top\bfX)]$ is finite and continuous, hence uniformly bounded on the compact set $\mathcal{B}$, and the same holds for the tilted moments $\E[\exp(\bfbeta^\top\bfX)\ind(T \ge t)]$ and $\E[\|\bfX\|^k\exp(\bfbeta^\top\bfX)\ind(T \ge t)]$ uniformly on $\mathcal{B}\times[0,\tau]$; $s_X^{(0)}(\bfbeta, t)$ and $\Prob(T \ge t)$ are bounded below uniformly on $[0, \tau]$ by (A1)(iii). 
Expectations of $\|\bfX\|^k$ and $\|\boldsymbol{\eta}(\bfX)\|^k$ are finite; and, since under (A3) the calibration residual $\boldsymbol{\epsilon}_m = \bfX - \mathbf{m}(\bfXs)$ is a linear combination of the jointly sub-Gaussian $(\bfX, \bfXs)$, it is itself sub-Gaussian. 

\begin{proof}
We proceed in three steps.

\emph{Step 1: Equivalence of naive and regression calibration estimators.}
The \emph{regression calibration (RC) estimator} replaces $\bfXs$ with $\mathbf{m}(\bfXs) = \mathbf{a} + \mathbf{B}\bfXs$ in the partial likelihood and maximizes over $\bfbeta$. The partial likelihood contribution at event time $T_i$ under RC is
\begin{equation*}
    \frac{\exp(\bfbeta^\top\mathbf{m}(\bfXs_i))}{\sum_{j \in \mathcal{R}(T_i)}\exp(\bfbeta^\top\mathbf{m}(\bfXs_j))}.
\end{equation*}
Expanding the exponent: $\bfbeta^\top\mathbf{m}(\bfXs_i) = \bfbeta^\top(\mathbf{a} + \mathbf{B}\bfXs_i) = \bfbeta^\top\mathbf{a} + (\mathbf{B}^\top\bfbeta)^\top\bfXs_i$. Since $\exp(\bfbeta^\top\mathbf{a})$ appears in every term of both numerator and denominator, it cancels:
\begin{equation*}
    \frac{\exp(\bfbeta^\top\mathbf{m}(\bfXs_i))}{\sum_{j \in \mathcal{R}(T_i)}\exp(\bfbeta^\top\mathbf{m}(\bfXs_j))} = \frac{\exp((\mathbf{B}^\top\bfbeta)^\top\bfXs_i)}{\sum_{j \in \mathcal{R}(T_i)}\exp((\mathbf{B}^\top\bfbeta)^\top\bfXs_j)}.
\end{equation*}
The right-hand side is the naive partial likelihood with covariates $\bfXs$ at coefficient $\mathbf{B}^\top\bfbeta$. Since the map $\bfbeta \mapsto \mathbf{B}^\top\bfbeta$ is a bijection (assuming $\mathbf{B}$ is invertible), the maximizers are related by
\begin{equation*}
    \hat{\bfbeta}^* = \mathbf{B}^\top\hat{\bfbeta}_{\mathrm{RC}}.
\end{equation*}
The bias problem therefore reduces to showing that $\hat{\bfbeta}_{\mathrm{RC}}$ approximately targets $\bfbeta$.

\emph{Step 2: Approximate consistency of the RC estimator.}
The population RC score at $\bfbeta$ is
\begin{equation*}
    \mathbf{u}_{\mathrm{RC}}(\bfbeta) = \E\Bigl[\Delta\bigl\{\mathbf{m}(\bfXs) - \bar{\mathbf{m}}(\bfbeta, T)\bigr\}\Bigr],
\end{equation*}
where
\begin{equation*}
    \bar{\mathbf{m}}(\bfbeta, t) = \frac{\E[\mathbf{m}(\bfXs)\exp(\bfbeta^\top\mathbf{m}(\bfXs))\ind(T \ge t)]}{\E[\exp(\bfbeta^\top\mathbf{m}(\bfXs))\ind(T \ge t)]}
\end{equation*}
is the risk-set weighted average of the calibrated covariates. We need to show that $\mathbf{u}_{\mathrm{RC}}(\bfbeta) \approx \mathbf{0}$.

Add and subtract $\bfX$ and $\bar{\bfX}(\bfbeta, T)$ (the true risk-set mean) to decompose:
\begin{align*}
    \mathbf{u}_{\mathrm{RC}}(\bfbeta) &= \underbrace{\E\bigl[\Delta\bigl\{\mathbf{m}(\bfXs) - \bfX\bigr\}\bigr]}_{\text{Term A: double-projection}}\\
    &\quad + \underbrace{\E\bigl[\Delta\bigl\{\bfX - \bar{\bfX}(\bfbeta, T)\bigr\}\bigr]}_{\text{Term B: true Cox score}}\\
    &\quad + \underbrace{\E\bigl[\Delta\bigl\{\bar{\bfX}(\bfbeta, T) - \bar{\mathbf{m}}(\bfbeta, T)\bigr\}\bigr]}_{\text{Term C: risk-set weight discrepancy}}.
\end{align*}
Term B is the true Cox score $\mathbf{u}(\bfbeta) = \mathbf{0}$ by definition. We bound Terms A and C.

\emph{Step 2a: The double-projection residual.}
Under non-differential error (A2), $\bfXs \perp (\Delta, T) \mid \bfX$, so for any function $h(\bfXs)$:
\begin{equation*}
    \E[\Delta \cdot h(\bfXs)] = \E\bigl[\E[\Delta \mid \bfX]\cdot\E[h(\bfXs) \mid \bfX]\bigr].
\end{equation*}
Applying with $h(\bfXs) = \mathbf{m}(\bfXs) - \bfX$:
\begin{equation*}
    \E\bigl[\Delta\bigl\{\mathbf{m}(\bfXs) - \bfX\bigr\}\bigr] = \E\bigl[\E[\Delta \mid \bfX]\cdot\bigl(\E[\mathbf{m}(\bfXs) \mid \bfX] - \bfX\bigr)\bigr].
\end{equation*}
Define the \emph{double-projection residual}
\begin{equation}
    \boldsymbol{\eta}(\bfX) := \E[\mathbf{m}(\bfXs) \mid \bfX] - \bfX.
    \label{eq:double_proj}
\end{equation}

This is the discrepancy between $\bfX$ and its double projection $\E[\E[\bfX \mid \bfXs] \mid \bfX]$. By the law of total expectation, $\E[\boldsymbol{\eta}(\bfX)] = \E[\E[\bfX \mid \bfXs]] - \E[\bfX] = \mathbf{0}$.

Under (A3), $\mathbf{m}(\mathbf{x}^*) = \mathbf{a} + \mathbf{B}\mathbf{x}^*$ with $\mathbf{a} = \boldsymbol{\mu}_X - \mathbf{B}\boldsymbol{\mu}_{X^*}$:
\begin{align*}
    \E[\mathbf{m}(\bfXs) \mid \bfX] &= \mathbf{a} + \mathbf{B}\,\E[\bfXs \mid \bfX] \notag\\
    &= \boldsymbol{\mu}_X - \mathbf{B}\boldsymbol{\mu}_{X^*} + \mathbf{B}\,\E[\bfXs \mid \bfX],
\end{align*}
so
\begin{equation*}
    \boldsymbol{\eta}(\bfX) = \mathbf{B}\bigl(\E[\bfXs \mid \bfX] - \boldsymbol{\mu}_{X^*}\bigr) - (\bfX - \boldsymbol{\mu}_X).
\end{equation*}

We now show that $\Var(\boldsymbol{\eta}(\bfX))$ is controlled by $\E[\boldsymbol{\Sigma}_{X|X^*}]$. 
Define the calibration residual $\boldsymbol{\epsilon}_m := \bfX - \mathbf{m}(\bfXs)$ (the difference between the truth and the calibration prediction). 
Note that 
$$
\E[\boldsymbol{\epsilon}_m] = \E[\bfX] - \E[\mathbf{m}(\bfXs)] = \E[\bfX] - \E[\E[\bfX \mid \bfXs]] = \mathbf{0}
$$ 
and 
$$
\E[\boldsymbol{\epsilon}_m \mid \bfX] = \bfX - \E[\mathbf{m}(\bfXs) \mid \bfX] = -\boldsymbol{\eta}(\bfX).
$$

Then, the law of total variance gives 
\begin{align*}
\Var(\boldsymbol{\epsilon}_m) 
&= \Var(\E[\boldsymbol{\epsilon}_m \mid \bfX]) + \E[\Var(\boldsymbol{\epsilon}_m \mid \bfX)] \\
&= \Var(-\boldsymbol{\eta}(\bfX)) + \E[\Var(\boldsymbol{\epsilon}_m \mid \bfX)]
\end{align*}
Taking traces and using $\E[\boldsymbol{\epsilon}_m] = \mathbf{0}$:
\begin{equation*}
    \underbrace{\E[\|\boldsymbol{\epsilon}_m\|^2]}_{\mathrm{tr}(\Var(\boldsymbol{\epsilon}_m))} = \underbrace{\E[\|\boldsymbol{\eta}(\bfX)\|^2]}_{\mathrm{tr}(\Var(\boldsymbol{\eta}))} + \underbrace{\E[\E[\|\boldsymbol{\epsilon}_m - \E[\boldsymbol{\epsilon}_m \mid \bfX]\|^2 \mid \bfX]]}_{\mathrm{tr}(\E[\Var(\boldsymbol{\epsilon}_m \mid \bfX)]) \;\ge\; 0},
\end{equation*}
where the first equality on the left uses $\mathrm{tr}(\Var(Y)) = \E[\|Y\|^2]$ when $\E[Y] = \mathbf{0}$, and similarly for $\boldsymbol{\eta}$.

The left hand side equals $\mathrm{tr}(\E[\boldsymbol{\Sigma}_{X|X^*}])$ since $\E[\|\bfX - \mathbf{m}(\bfXs)\|^2] = \mathrm{tr}(\E[(\bfX - \mathbf{m}(\bfXs))(\bfX - \mathbf{m}(\bfXs))^\top])  = \mathrm{tr}(\E[\boldsymbol{\Sigma}_{X|X^*}])$. 
Dropping the non-negative third term:
\begin{equation}
    \mathrm{tr}(\Var(\boldsymbol{\eta}(\bfX))) \le \mathrm{tr}(\E[\boldsymbol{\Sigma}_{X|X^*}]).
    \label{eq:eta_bound}
\end{equation}

Thus, the variance of the double-projection residual is bounded by the calibration residual variance.

Since $\E[\boldsymbol{\eta}(\bfX)] = \mathbf{0}$, Term A becomes a covariance:
\begin{equation}
    \text{Term A} = \Cov\bigl(\E[\Delta \mid \bfX],\;\boldsymbol{\eta}(\bfX)\bigr).
    \label{eq:term_a_cov}
\end{equation}
We bound Terms A and C by decomposing each into a $\bfbeta$-independent censoring contribution and a $\bfbeta$-dependent hazard contribution.

\emph{Step 2b: Decomposing Term A.}
Under the Cox model with censoring survival $G(t \mid \bfX) = \Prob(C \ge t \mid \bfX)$,
\begin{equation*}
    \E[\Delta \mid \bfX] = \int_0^\tau \lambda_0(t)\exp(\bfbeta^\top\bfX)\,S_0(t)^{\exp(\bfbeta^\top\bfX)}\,G(t \mid \bfX)\,dt.
\end{equation*}
When censoring depends on covariates, $\E[\Delta \mid \bfX]$ depends on $\bfX$ through both the event hazard ($\bfbeta$-dependent) and the censoring mechanism ($\bfbeta$-independent). Write
\begin{equation*}
    \E[\Delta \mid \bfX] = \underbrace{h(\bfX)}_{\text{censoring contribution}} + \underbrace{g(\bfX, \bfbeta)}_{\text{hazard contribution}},
\end{equation*}
where $h(\bfX) := \E[\Delta \mid \bfX]\big|_{\bfbeta = \mathbf{0}} = \int_0^\tau \lambda_0(t)\,S_0(t)\,G(t \mid \bfX)\,dt$ is the event probability under the null model (which depends on $\bfX$ only through censoring), and $g(\bfX, \bfbeta) = \E[\Delta \mid \bfX] - h(\bfX)$ captures the effect of the covariates on the event hazard. By construction, $g(\bfX, \mathbf{0}) = 0$.
Moreover $\E[\Delta\mid\bfX]$ is smooth in $\bfbeta$ with gradient $\nabla_{\bfbeta}\E[\Delta\mid\bfX]\big|_{\bfbeta=\mathbf{0}} = \bfX\,\phi(\bfX)$, where $\phi(\bfX) = \int_0^\tau \lambda_0(t)\,S_0(t)\,\bigl(1 + \log S_0(t)\bigr)G(t\mid\bfX)\,dt$ is bounded on $[0,\tau]$ because the survival-damped integrand $u\,S_0(t)^u$ is bounded over $u \ge 0$ and $\log S_0(t)$ is bounded by (A1)(iii).
Hence, by a first-order expansion $g(\bfX, \bfbeta) = O(\|\bfbeta\|)$ in $L^2$ and $\Var(g(\bfX, \bfbeta)) = O(\|\bfbeta\|^2)$, with the constant governed by $\E\|\bfX\|^2 < \infty$.

Substituting into \eqref{eq:term_a_cov}:
\begin{equation*}
    \text{Term A} = \underbrace{\Cov(h(\bfX),\,\boldsymbol{\eta}(\bfX))}_{\text{A}_0\text{: censoring--calibration}} + \underbrace{\Cov(g(\bfX, \bfbeta),\,\boldsymbol{\eta}(\bfX))}_{\text{A}_\beta\text{: hazard--calibration}}.
\end{equation*}
The hazard--calibration term $\text{A}_\beta$ is $O(\|\bfbeta\| \cdot \|\E[\boldsymbol{\Sigma}_{X|X^*}]\|^{1/2})$ by Cauchy--Schwarz and \eqref{eq:eta_bound}, since $\sqrt{\Var(g)} = O(\|\bfbeta\|)$ and $\sqrt{\mathrm{tr}(\Var(\boldsymbol{\eta}))} \le \sqrt{\mathrm{tr}(\E[\boldsymbol{\Sigma}_{X|X^*}])}$. 
The censoring--calibration term $\text{A}_0 = \Cov(h(\bfX), \boldsymbol{\eta}(\bfX))$ does not involve $\bfbeta$ and is generally nonzero when censoring depends on covariates.

\emph{Step 2c: Decomposing Term C.}
A parallel decomposition applies to Term C. The risk-set weighted averages are:
\begin{align*}
    \bar{\bfX}(\bfbeta, t) &= \frac{s_X^{(1)}(\bfbeta, t)}{s_X^{(0)}(\bfbeta, t)}, &\quad s_X^{(k)}(\bfbeta, t) &= \E[\bfX^{\otimes k}\exp(\bfbeta^\top\bfX)\ind(T \ge t)],\\
    \bar{\mathbf{m}}(\bfbeta, t) &= \frac{s_m^{(1)}(\bfbeta, t)}{s_m^{(0)}(\bfbeta, t)}, &\quad s_m^{(k)}(\bfbeta, t) &= \E[\mathbf{m}(\bfXs)^{\otimes k}\exp(\bfbeta^\top\mathbf{m}(\bfXs))\ind(T \ge t)],
\end{align*}
and we define the risk-set discrepancy $\mathbf{D}(\bfbeta, t) := \bar{\bfX}(\bfbeta, t) - \bar{\mathbf{m}}(\bfbeta, t)$, so
\begin{equation}
    \text{Term C} = \E\bigl[\Delta\,\mathbf{D}(\bfbeta, T)\bigr].
    \label{eq:TC_def}
\end{equation}

Recall the calibration residual $\boldsymbol{\epsilon}_m = \bfX - \mathbf{m}(\bfXs)$ from Step~2a, so that $\mathbf{m}(\bfXs) = \bfX - \boldsymbol{\epsilon}_m$ and by a Taylor expansion:
\begin{equation*}
    \exp(\bfbeta^\top\mathbf{m}(\bfXs)) = \exp(\bfbeta^\top\bfX)\cdot\exp(-\bfbeta^\top\boldsymbol{\epsilon}_m) = \exp(\bfbeta^\top\bfX)\bigl[1 - \bfbeta^\top\boldsymbol{\epsilon}_m\bigr] + O(\|\bfbeta\|^2\|\boldsymbol{\epsilon}_m\|^2).
\end{equation*}

Substituting into $s_m^{(0)}$ and $s_m^{(1)}$, and using $\boldsymbol{\epsilon}_m \perp T \mid \bfX$ from (A2) together with $\E[\boldsymbol{\epsilon}_m \mid \bfX] = -\boldsymbol{\eta}(\bfX)$:
\begin{align}
    s_m^{(0)}(\bfbeta, t) &= s_X^{(0)}(\bfbeta, t) + \E\bigl[\exp(\bfbeta^\top\bfX)\,\bfbeta^\top\boldsymbol{\eta}(\bfX)\,\ind(T \ge t)\bigr] + R_0(\bfbeta, t), \label{eq:smk_expand0}\\
    s_m^{(1)}(\bfbeta, t) &= \E\bigl[(\bfX - \boldsymbol{\epsilon}_m)\exp(\bfbeta^\top\bfX)(1 - \bfbeta^\top\boldsymbol{\epsilon}_m)\ind(T \ge t)\bigr] + R_1(\bfbeta, t) \notag\\
    &= s_X^{(1)}(\bfbeta, t) + \E\bigl[\bfX\exp(\bfbeta^\top\bfX)\,\bfbeta^\top\boldsymbol{\eta}(\bfX)\,\ind(T \ge t)\bigr] \notag\\
    &\quad + \E\bigl[\boldsymbol{\eta}(\bfX)\exp(\bfbeta^\top\bfX)\,\ind(T \ge t)\bigr] + R_1(\bfbeta, t), \label{eq:smk_expand1}
\end{align}
where $\|R_0\|, \|R_1\| = O(\|\bfbeta\|^2\,\|\E[\boldsymbol{\Sigma}_{X|X^*}]\|)$.

Write $s_m^{(k)} = s_X^{(k)} + \Delta_k$, so that from \eqref{eq:smk_expand0}--\eqref{eq:smk_expand1}:
\begin{align}
    \Delta_0(\bfbeta, t) &= \E\bigl[\exp(\bfbeta^\top\bfX)\,\bfbeta^\top\boldsymbol{\eta}(\bfX)\,\ind(T \ge t)\bigr] + R_0(\bfbeta, t), \label{eq:Delta0_explicit}\\
    \Delta_1(\bfbeta, t) &= \underbrace{\E\bigl[\boldsymbol{\eta}(\bfX)\exp(\bfbeta^\top\bfX)\,\ind(T \ge t)\bigr]}_{=:\; \Delta_1^{\mathrm{(a)}}(\bfbeta, t)} + \underbrace{\E\bigl[\bfX\exp(\bfbeta^\top\bfX)\,\bfbeta^\top\boldsymbol{\eta}(\bfX)\,\ind(T \ge t)\bigr]}_{=:\; \Delta_1^{\mathrm{(b)}}(\bfbeta, t)} + R_1(\bfbeta, t). \label{eq:Delta1_explicit}
\end{align}

Note that $\|\E[\boldsymbol{\eta}(\bfX)\,\omega(\bfX, t)]\| \le \E[\|\boldsymbol{\eta}(\bfX)\|\cdot|\omega(\bfX, t)|]$ is bounded via Cauchy--Schwarz by $\sqrt{\E[\|\boldsymbol{\eta}\|^2]}\cdot\sqrt{\E[\omega^2]}$, and $\sqrt{\E[\|\boldsymbol{\eta}(\bfX)\|^2]} = O(\|\E[\boldsymbol{\Sigma}_{X|X^*}]\|^{1/2})$ by \eqref{eq:eta_bound}. Applying this:
\begin{enumerate}
    \item[(i)] \emph{$\Delta_0$ carries an explicit $\bfbeta^\top\boldsymbol{\eta}$ factor.} From \eqref{eq:Delta0_explicit}, $\Delta_0(\bfbeta, t) = \E[\exp(\bfbeta^\top\bfX)\,\bfbeta^\top\boldsymbol{\eta}(\bfX)\,\ind(T \ge t)] + R_0$. The leading term is bounded by $\|\bfbeta\|\cdot\E[\exp(\bfbeta^\top\bfX)\,\|\boldsymbol{\eta}(\bfX)\|\,\ind(T \ge t)] \le \|\bfbeta\|\cdot\sqrt{\E[\exp(2\bfbeta^\top\bfX)\ind(T \ge t)]}\cdot\sqrt{\E[\|\boldsymbol{\eta}\|^2]}$ by Cauchy--Schwarz, and $R_0 = O(\|\bfbeta\|^2\,\|\E[\boldsymbol{\Sigma}_{X|X^*}]\|)$ is strictly subdominant. Therefore
    \begin{equation}
        \|\Delta_0(\bfbeta, t)\| = O\bigl(\|\bfbeta\|\cdot\|\E[\boldsymbol{\Sigma}_{X|X^*}]\|^{1/2}\bigr),
        \label{eq:Delta0_order}
    \end{equation}
    and in particular $\Delta_0(\mathbf{0}, t) = 0$.

    \item[(ii)] \emph{$\Delta_1$ has a nonzero $\bfbeta$-independent piece.} From \eqref{eq:Delta1_explicit}, $\Delta_1^{\mathrm{(a)}}(\mathbf{0}, t) = \E[\boldsymbol{\eta}(\bfX)\ind(T \ge t)] \ne \mathbf{0}$ in general, with $\|\Delta_1^{\mathrm{(a)}}(\mathbf{0}, t)\| = O(\|\E[\boldsymbol{\Sigma}_{X|X^*}]\|^{1/2})$ by Cauchy--Schwarz and \eqref{eq:eta_bound}. The same bound holds for $\Delta_1^{\mathrm{(a)}}(\bfbeta, t)$ at any $\bfbeta$ (using boundedness of $\exp(\bfbeta^\top\bfX)$ on the support). The piece $\Delta_1^{\mathrm{(b)}}$ carries an explicit $\bfbeta^\top$ factor, so $\|\Delta_1^{\mathrm{(b)}}(\bfbeta, t)\| = O(\|\bfbeta\|\cdot\|\E[\boldsymbol{\Sigma}_{X|X^*}]\|^{1/2})$. Combining:
    \begin{equation}
        \|\Delta_1(\bfbeta, t)\| = O\bigl(\|\E[\boldsymbol{\Sigma}_{X|X^*}]\|^{1/2}\bigr), \qquad \|\Delta_1(\bfbeta, t) - \Delta_1(\mathbf{0}, t)\| = O\bigl(\|\bfbeta\|\cdot\|\E[\boldsymbol{\Sigma}_{X|X^*}]\|^{1/2}\bigr).
        \label{eq:Delta1_order}
    \end{equation}
\end{enumerate}

At $\bfbeta = \mathbf{0}$, $\Delta_0(\mathbf{0}, t) = 0$ but $\Delta_1(\mathbf{0}, t) \ne \mathbf{0}$. So the risk-set discrepancy at $\bfbeta = \mathbf{0}$ is
\begin{equation}
    \mathbf{D}(\mathbf{0}, t) = \bar{\bfX}(\mathbf{0}, t) - \bar{\mathbf{m}}(\mathbf{0}, t) = -\frac{\Delta_1(\mathbf{0}, t)}{s_X^{(0)}(\mathbf{0}, t)} = -\frac{\E[\boldsymbol{\eta}(\bfX)\ind(T \ge t)]}{\Prob(T \ge t)},
    \label{eq:D_zero}
\end{equation}
which is $\bfbeta$-independent, nonzero in general, and bounded by $\|\mathbf{D}(\mathbf{0}, t)\| = O(\|\E[\boldsymbol{\Sigma}_{X|X^*}]\|^{1/2})$ via Cauchy--Schwarz and \eqref{eq:eta_bound}.

Take $ \bar{\mathbf{m}}(\bfbeta, t) = \cfrac{s_X^{(1)} + \Delta_1}{s_X^{(0)} + \Delta_0}$.
The ratio expansion uses $1/(a + x) = 1/a - x/a^2 + x^2/a^3 + O(x^3)$ applied to $1/(s_X^{(0)}(\bfbeta,t) + \Delta_0(\bfbeta,t))$:
\begin{align*}
    \bar{\mathbf{m}}(\bfbeta, t) &= \bigl(s_X^{(1)} + \Delta_1\bigr)\cdot\left[\frac{1}{s_X^{(0)}} - \frac{\Delta_0}{(s_X^{(0)})^2} + O\bigl(\Delta_0^2/(s_X^{(0)})^3\bigr)\right]\\
    &= \bar{\bfX}(\bfbeta, t) + \frac{\Delta_1(\bfbeta, t) - \bar{\bfX}(\bfbeta, t)\,\Delta_0(\bfbeta, t)}{s_X^{(0)}(\bfbeta, t)} \;\underbrace{- \frac{\Delta_1(\bfbeta, t)\,\Delta_0(\bfbeta, t)}{(s_X^{(0)}(\bfbeta, t))^2} + O\bigl(\|\Delta_0\|^2\bigr)}_{\text{second-order remainder}},
\end{align*}
so
\begin{equation*}
    \mathbf{D}(\bfbeta, t) = -\frac{\Delta_1(\bfbeta, t) - \bar{\bfX}(\bfbeta, t)\,\Delta_0(\bfbeta, t)}{s_X^{(0)}(\bfbeta, t)} + \frac{\Delta_1(\bfbeta, t)\,\Delta_0(\bfbeta, t)}{(s_X^{(0)}(\bfbeta, t))^2} + O(\|\Delta_0\|^2).
\end{equation*}
The second-order remainder contains a $\Delta_1\Delta_0$ cross term and a $\Delta_0^2$ piece. Using \eqref{eq:Delta0_order} and \eqref{eq:Delta1_order}:
\begin{align*}
    \|\Delta_1\Delta_0/(s_X^{(0)})^2\| &= O\bigl(\|\bfbeta\|\cdot\|\E[\boldsymbol{\Sigma}_{X|X^*}]\|\bigr),\\
    \|\Delta_0^2/(s_X^{(0)})^3\| &= O\bigl(\|\bfbeta\|^2\cdot\|\E[\boldsymbol{\Sigma}_{X|X^*}]\|\bigr).
\end{align*}
Both are subdominant to the leading-order bound $O(\|\bfbeta\|\cdot\|\E[\boldsymbol{\Sigma}_{X|X^*}]\|^{1/2})$ established below.

Subtracting $\mathbf{D}(\mathbf{0}, t) = -\Delta_1(\mathbf{0}, t)/\Prob(T \ge t)$ from $\mathbf{D}(\bfbeta, t)$:
\begin{align}
    \mathbf{D}(\bfbeta, t) - \mathbf{D}(\mathbf{0}, t) &= \underbrace{\Delta_1(\mathbf{0}, t)\left[\frac{1}{\Prob(T\ge t)} - \frac{1}{s_X^{(0)}(\bfbeta, t)}\right]}_{\text{(I)}} \notag\\
    &\quad - \underbrace{\frac{\Delta_1(\bfbeta, t) - \Delta_1(\mathbf{0}, t)}{s_X^{(0)}(\bfbeta, t)}}_{\text{(II)}} + \underbrace{\frac{\bar{\bfX}(\bfbeta, t)\,\Delta_0(\bfbeta, t)}{s_X^{(0)}(\bfbeta, t)}}_{\text{(III)}} + \text{(subdominant)}. \label{eq:D_diff}
\end{align}
We bound each piece. For (I): $s_X^{(0)}(\bfbeta, t) - \Prob(T \ge t) = \E[(\exp(\bfbeta^\top\bfX) - 1)\ind(T \ge t)] = O(\|\bfbeta\|)$ by Taylor expansion, so the bracket is $O(\|\bfbeta\|)$; combined with $\|\Delta_1(\mathbf{0}, t)\| = O(\|\E[\boldsymbol{\Sigma}_{X|X^*}]\|^{1/2})$ from \eqref{eq:Delta1_order}, this gives $\|(\text{I})\| = O(\|\bfbeta\|\cdot\|\E[\boldsymbol{\Sigma}_{X|X^*}]\|^{1/2})$.

For (II): the second bound in \eqref{eq:Delta1_order} gives directly $\|\Delta_1(\bfbeta, t) - \Delta_1(\mathbf{0}, t)\| = O(\|\bfbeta\|\cdot\|\E[\boldsymbol{\Sigma}_{X|X^*}]\|^{1/2})$, and $s_X^{(0)}(\bfbeta, t)$ is bounded below by regularity, so $\|(\text{II})\| = O(\|\bfbeta\|\cdot\|\E[\boldsymbol{\Sigma}_{X|X^*}]\|^{1/2})$.

For (III): \eqref{eq:Delta0_order} gives $\|\Delta_0(\bfbeta, t)\| = O(\|\bfbeta\|\cdot\|\E[\boldsymbol{\Sigma}_{X|X^*}]\|^{1/2})$, and $\|\bar{\bfX}(\bfbeta, t)\|$ is bounded on the support, so $\|(\text{III})\| = O(\|\bfbeta\|\cdot\|\E[\boldsymbol{\Sigma}_{X|X^*}]\|^{1/2})$.

Combining:
\begin{equation}
    \mathbf{D}(\bfbeta, t) = \mathbf{D}(\mathbf{0}, t) + \boldsymbol{\rho}(\bfbeta, t), \qquad \|\boldsymbol{\rho}(\bfbeta, t)\| = O\bigl(\|\bfbeta\|\cdot\|\E[\boldsymbol{\Sigma}_{X|X^*}]\|^{1/2}\bigr),
    \label{eq:D_bound}
\end{equation}
so $\mathbf{D}(\bfbeta, t)$ splits cleanly into a $\bfbeta$-independent piece $\mathbf{D}(\mathbf{0}, t)$ of size $O(\|\E[\boldsymbol{\Sigma}_{X|X^*}]\|^{1/2})$ and a $\bfbeta$-dependent remainder $\boldsymbol{\rho}(\bfbeta, t)$ of size $O(\|\bfbeta\|\cdot\|\E[\boldsymbol{\Sigma}_{X|X^*}]\|^{1/2})$.

\emph{Step 2c(ii): Decomposing Term C via the censoring/hazard split.} Substituting \eqref{eq:D_bound} into \eqref{eq:TC_def} and using iterated expectations with $\E[\Delta \mid \bfX] = h(\bfX) + g(\bfX, \bfbeta)$ from Step~2b:
\begin{align*}
    \text{Term C} &= \E\bigl[\Delta\,\mathbf{D}(\mathbf{0}, T)\bigr] + \E\bigl[\Delta\,\boldsymbol{\rho}(\bfbeta, T)\bigr]\\
    &= \underbrace{\E\bigl[h(\bfX)\,\mathbf{D}(\mathbf{0}, T)\bigr]}_{=:\;\text{C}_0} + \underbrace{\E\bigl[g(\bfX, \bfbeta)\,\mathbf{D}(\mathbf{0}, T)\bigr] + \E\bigl[\Delta\,\boldsymbol{\rho}(\bfbeta, T)\bigr]}_{=:\;\text{C}_\beta}.
\end{align*}
By construction, $\text{C}_0$ is $\bfbeta$-independent (since $\mathbf{D}(\mathbf{0}, t)$ depends only on baseline risk-set composition, not on $\bfbeta$). The two terms defining $\text{C}_\beta$ are each $O(\|\bfbeta\|\cdot\|\E[\boldsymbol{\Sigma}_{X|X^*}]\|^{1/2})$: the first by Cauchy--Schwarz with $\sqrt{\Var(g(\bfX, \bfbeta))} = O(\|\bfbeta\|)$ and $\sup_t\|\mathbf{D}(\mathbf{0}, t)\|$ controlled by $\|\E[\boldsymbol{\Sigma}_{X|X^*}]\|^{1/2}$ via \eqref{eq:eta_bound} applied at $\bfbeta = \mathbf{0}$; the second directly from \eqref{eq:D_bound} since $|\Delta| \le 1$. Therefore
\begin{equation*}
    \|\text{C}_\beta\| = O\bigl(\|\bfbeta\|\cdot\|\E[\boldsymbol{\Sigma}_{X|X^*}]\|^{1/2}\bigr),
\end{equation*}
matching the rate of $\text{A}_\beta$ from Step~2b.

\emph{Step 2d: Cancellation of censoring cross-terms.}
At $\bfbeta = \mathbf{0}$, $\mathbf{u}_{\mathrm{RC}}(\mathbf{0}) = \text{A}_0 + \text{C}_0$ (since $\mathbf{u}(\bfbeta) = \mathbf{0}$ and $\text{A}_\beta, \text{C}_\beta$ vanish). A martingale argument shows $\mathbf{u}_{\mathrm{RC}}(\mathbf{0}) = \mathbf{0}$: at $\bfbeta = \mathbf{0}$, the RC score is $\E[\int_0^\tau \{\mathbf{m}(\bfXs) - \bar{\mathbf{m}}(\mathbf{0}, t)\}\,dM(t)]$, where $M(t) = N(t) - \int_0^t Y(s)\lambda_0(s)\,ds$ is a martingale under the true null model. Since $\mathbf{m}(\bfXs) - \bar{\mathbf{m}}(\mathbf{0}, t)$ is predictable, this integral has zero expectation. Therefore $\text{A}_0 + \text{C}_0 = \mathbf{0}$: the censoring cross-terms from Terms A and C cancel exactly.

The surviving contributions are the $\bfbeta$-dependent parts $\text{A}_\beta$ and $\text{C}_\beta$, both $O(\|\bfbeta\| \cdot \|\E[\boldsymbol{\Sigma}_{X|X^*}]\|^{1/2})$. Therefore
\begin{equation}
    \|\mathbf{u}_{\mathrm{RC}}(\bfbeta)\| = \|\text{A}_\beta + \text{C}_\beta\| = O\bigl(\|\bfbeta\| \cdot \|\E[\boldsymbol{\Sigma}_{X|X^*}]\|^{1/2}\bigr).
    \label{eq:term_a_order}
\end{equation}
\emph{Step 2e: Taylor inversion of the score.}
From \eqref{eq:term_a_order}:
\begin{equation*}
    \mathbf{u}_{\mathrm{RC}}(\bfbeta) = O(\|\bfbeta\| \cdot \|\E[\boldsymbol{\Sigma}_{X|X^*}]\|^{1/2}).
\end{equation*}
Let $\bfbeta_{\mathrm{RC}}^*$ denote the population target of $\hat{\bfbeta}_{\mathrm{RC}}$ (where $\mathbf{u}_{\mathrm{RC}}(\bfbeta_{\mathrm{RC}}^*) = \mathbf{0}$). Taylor-expanding around $\bfbeta$:
\begin{equation*}
    \mathbf{0} = \mathbf{u}_{\mathrm{RC}}(\bfbeta_{\mathrm{RC}}^*) \approx \mathbf{u}_{\mathrm{RC}}(\bfbeta) + \frac{\partial \mathbf{u}_{\mathrm{RC}}}{\partial \bfbeta^\top}\bigg|_{\bfbeta}(\bfbeta_{\mathrm{RC}}^* - \bfbeta).
\end{equation*}
The derivative $\partial\mathbf{u}_{\mathrm{RC}}/\partial\bfbeta^\top|_{\bfbeta}$ is $-\boldsymbol{\Omega}_{\mathrm{RC}}$, the negative RC information matrix. Solving:
\begin{equation*}
    \bfbeta_{\mathrm{RC}}^* - \bfbeta = \boldsymbol{\Omega}_{\mathrm{RC}}^{-1}\,\mathbf{u}_{\mathrm{RC}}(\bfbeta).
\end{equation*}
This is the RC residual: $\|\boldsymbol{\Omega}_{\mathrm{RC}}^{-1}\mathbf{u}_{\mathrm{RC}}(\bfbeta)\| = O(\|\bfbeta\| \cdot \|\E[\boldsymbol{\Sigma}_{X|X^*}]\|^{1/2})$.

\emph{Step 3: Combining via the equivalence.}
From Step~2, $\bfbeta_{\mathrm{RC}}^* = \bfbeta + \boldsymbol{\Omega}_{\mathrm{RC}}^{-1}\mathbf{u}_{\mathrm{RC}}(\bfbeta)$. From Step~1, $\bfbeta^* = \mathbf{B}^\top\bfbeta_{\mathrm{RC}}^*$, so
\begin{equation*}
    \bfbeta^* = \mathbf{B}^\top\bfbeta + \mathbf{B}^\top\boldsymbol{\Omega}_{\mathrm{RC}}^{-1}\mathbf{u}_{\mathrm{RC}}(\bfbeta).
\end{equation*}
Solving for $\bfbeta$:
\begin{equation*}
    \bfbeta = \bfbeta^* + \bigl[(\mathbf{B}^\top)^{-1} - \mathbf{I}\bigr]\bfbeta^* - \boldsymbol{\Omega}_{\mathrm{RC}}^{-1}\mathbf{u}_{\mathrm{RC}}(\bfbeta),
\end{equation*}
which gives \eqref{eq:bias_unified}. Dropping the RC residual gives the leading-order approximation \eqref{eq:bias_leading}.
\end{proof}

\clearpage

\section{Coverage convergence of CIs with $n_v$}
\label{app:prop-ci-cvg}

To assess how the two tiers of confidence interval respond to the amount of validation data, we reran the cross-dependent error design of Study~1 while varying the validation-sample size $n_v$ over the grid $\{50, 100, 200, 300, 400, 800, 1600, 3200\}$ and holding the study-sample size fixed at $n_{\text{study}} = 1500$. 
The data-generating process is otherwise identical.
We report empirical coverage across methods and $n_v$ size in Figure~\ref{fig:valsize}.

\begin{figure}[ht!]
\centering
\includegraphics[width=0.85\textwidth]{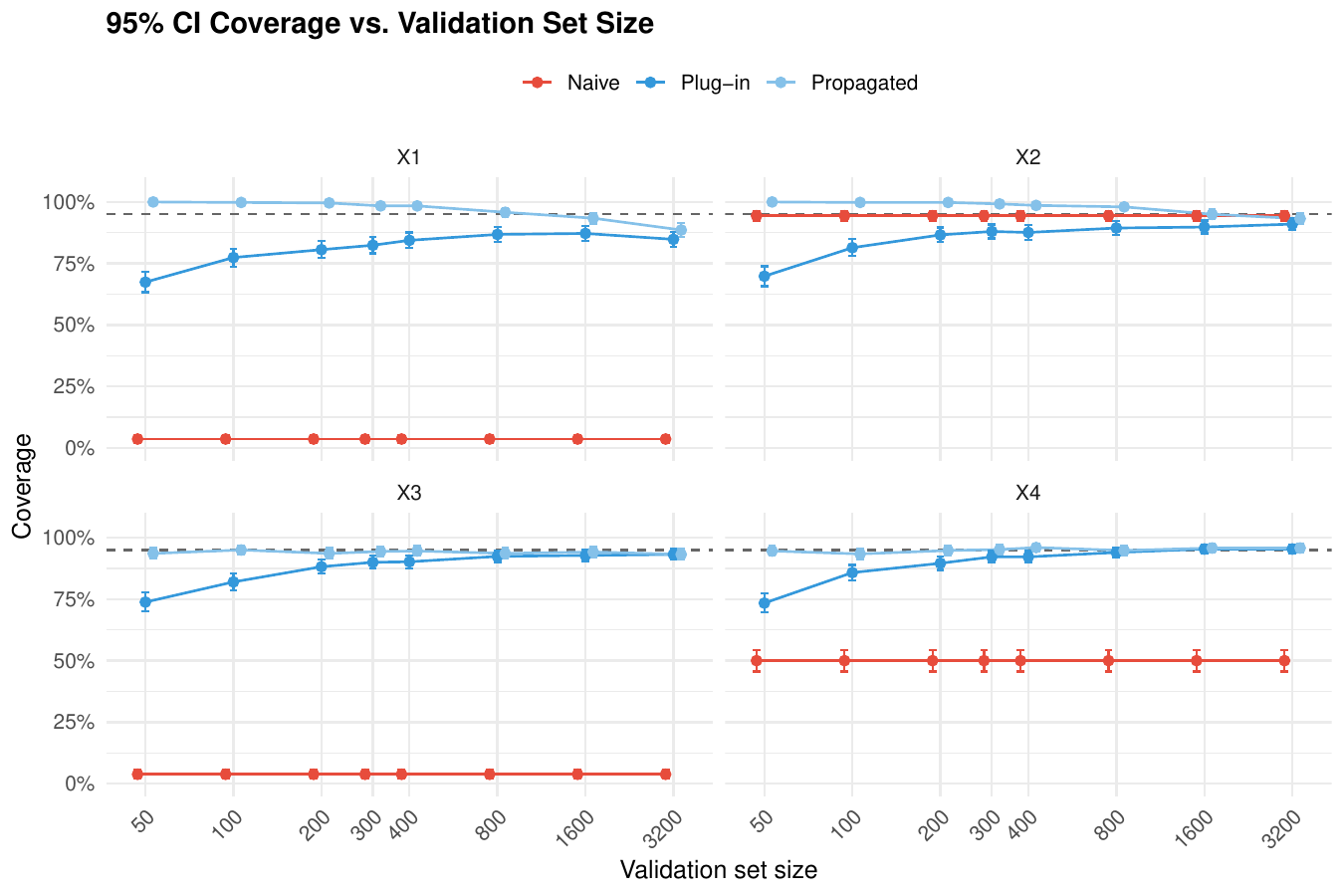}
\caption{95\% CI coverage vs $n_v$ for Study 1 data-generating process; error bars are 95\% Monte Carlo CIs.}
\label{fig:valsize}
\end{figure}

As expected, naive coverage is invariant to $n_v$, since it never uses the validation data: it remains badly miscalibrated for the covariates with appreciable attenuation and is near nominal for $X_2$ only because its naive bias is small. 
The plug-in interval undercovers when $n_v$ is small because it ignores the sampling variability of $\hat{\mathbf{B}}$, which is largest for small validation sets. 
Its coverage rises monotonically toward the nominal level as $n_v$ grows. 

The propagated intervals attain near-nominal coverage across the entire grid. 
At small $n_v$ it is conservative: for continuous covariates it
covers essentially $100\%$ of the time at $n_v = 50$ because the variance inflation that accounts for $\hat{\mathbf{B}}$ is largest exactly when the validation set is small; as $n_v$ increases this inflation shrinks and propagated coverage descends toward the nominal $95\%$. 
Consequently the gap between the plug-in and propagated intervals is widest at small $n_v$ and narrows as the two converge, quantifying the practical value of propagating calibration uncertainty precisely in the small $n_v$ regime. 
For the binary covariates the propagated interval tracks the nominal level closely throughout, remaining between roughly $93\%$ and $96\%$ for both $X_3$ and $X_4$.

\clearpage

\section{Semi-synthetic example using MIMIC-IV dataset}
\label{app:mimic}

Here, we analyze data derived from the Medical Information Mart for Intensive Care (MIMIC)-IV version 3.1 \citep{PhysioNet-mimiciv-3.1, johnson2023mimic, goldberger2000physiobank} on patients admitted to an intensive care unit.
We obtain baseline clinical factors and time to inpatient mortality, censoring surviving patients at the end of their stay. 
We aim to fit a Cox proportional hazards model with covariates contaminated by measurement error, and demonstrate the realistic use and performance of our bias correction.
The covariates analyzed are age at admission, sex, Sequential Organ Failure Assessment (SOFA) score, history of diabetes, and history of chronic obstructive pulmonary disease (COPD).
We leave age and sex as error-free measures, and simulate moderate error among the remaining covariates.
The SOFA score receives independent zero-mean additive Gaussian noise with standard deviation $\sigma_{\text{noise}} = 0.7\,\sigma_{\text{SOFA}}$; the corrupted values are rounded and clipped to the non-negative integers to respect SOFA's support. 
Diabetes and COPD status (binary variables) are corrupted by correlated bit-flips sharing a common latent factor $Z \sim \mathcal{N}(0,1)$: each has flip probability $p_k = \operatorname{logit}^{-1}(\alpha_k + \gamma_k |Z|)$, calibrated to a marginal flip rate of $\sim15\%$. 

We randomly sample a study dataset of $n_{\text{study}} = 5000$ and a disjoint validation dataset with $n_v = 500$ used to compute error summaries.
The validation set yields marginal error metrics visualized in Figure~\ref{fig:mimic-measurement-error}, and the calibration slope matrix used for bias correction:
$$
\hat{\mathbf{B}} = \bordermatrix{
              & \text{age} & \text{female} & \text{sofa} & \text{diab} & \text{copd} \cr
 \text{age}   &  1.0000 &  0.0000 &  0.0000 & 0.0000 & 0.0000 \cr
 \text{female}&  0.0000 &  1.0000 &  0.0000 & 0.0000 & 0.0000 \cr
 \text{sofa}  & -0.0007 & -0.1204 &  0.7430 & 0.1215 & 0.4018 \cr
 \text{diab}  &  0.0033 &  0.0140 &  0.0026 & 0.6351 & 0.0360 \cr
 \text{copd}  &  0.0036 &  0.0793 & -0.0009 & 0.0001 & 0.5470 \cr}
$$
While age and sex are error-free, their correlations with other covariates impact the resulting bias correction.
The condition number of $\hat{\mathbf{B}}$ is 2.28, indicating that the matrix inversion and bias correction should be stable.

\begin{figure}[htbp]
  \centering
  \includegraphics[width=\textwidth]{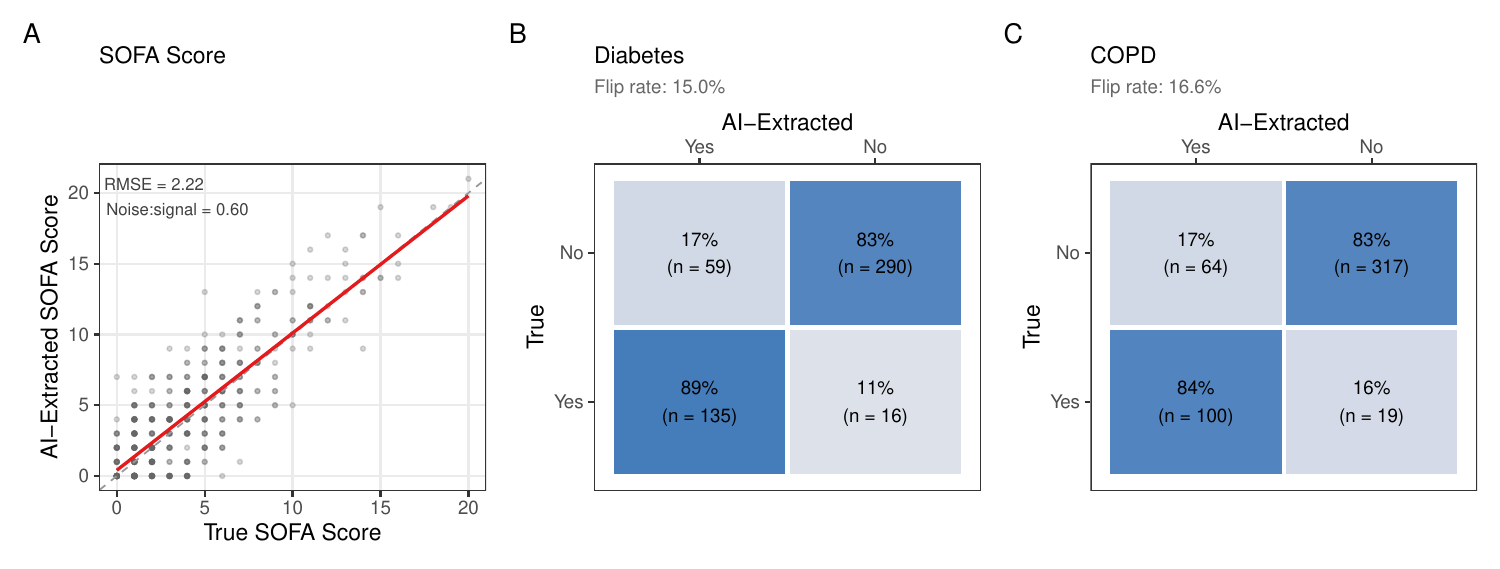}
  \caption{Marginal measurement error distributions for covariates in MIMIC-IV data experiment.}
\label{fig:mimic-measurement-error}
\end{figure}

The results of our analysis are visualized in Figure~\ref{fig:mimic-results}.
We compute the ground truth model parameters by fitting a Cox model on error-free covariates in study data.
The naive model fit results in moderate bias for HR estimates, most critically for SOFA score, where the naive confidence interval does not cover the ground truth parameter.
The corrected parameters uniformly reduce bias relative to the ground truth.
With the exception of SOFA score, the widths of plug-in confidence intervals are similar to those of the propagated CIs, indicating that the variance elements for these parameters are reasonably estimated from the validation set.

\begin{figure}[htbp]
  \centering
  \includegraphics[width=0.85\textwidth]{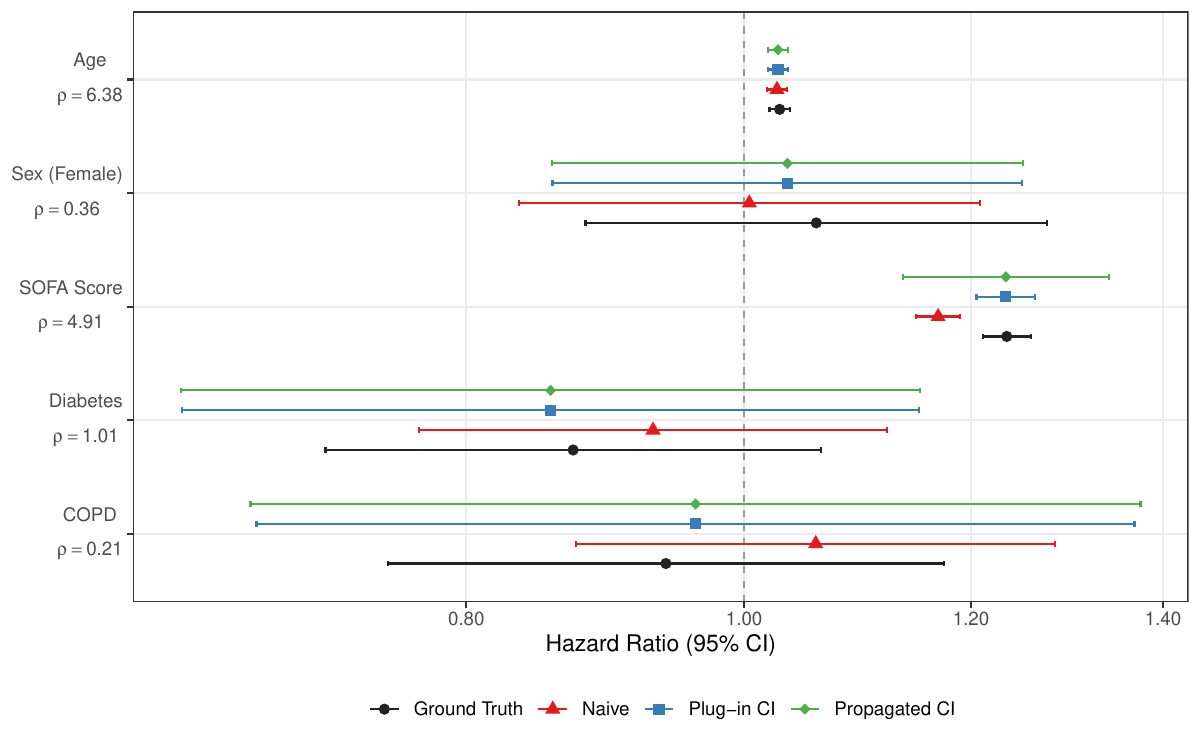}
  \caption{Estimated naive and bias-corrected hazard ratios with 95\% confidence intervals in MIMIC-IV data experiment.}
\label{fig:mimic-results}
\end{figure}

We also report the sensitivity diagnostic $\rho$ for each covariate.
We observe that the sex and COPD coefficients are least susceptible to bias due to the RC residual term, while those for age, SOFA score, and diabetes may be more sensitive.
While our ground truth comparisons in this example show reasonably low bias among all corrected estimates, in practice, a large $\rho$ value could indicate unreliability for larger extraction error settings.

\end{document}